\pdfoutput=1

\documentclass[a4paper]{article}
\RequirePackage{etex}

\usepackage[left=1.5cm,right=1.5cm,top=2cm,bottom=2cm,ignoreheadfoot]{geometry}

\usepackage[utf8]{inputenc}
\usepackage{pgfplots}
\pgfplotsset{compat=newest}
\usepackage[title]{appendix}
\usepackage{amssymb}
\usepackage{booktabs}
\usepackage{textcomp}
\usepackage{graphicx}
\usepackage{physics}
\usepackage{mathtools}
\usepackage[export]{adjustbox}
\usepackage{yhmath}
\usepackage{makecell}
\usepackage{amsmath}
\usepackage{dsfont}
\usepackage{wrapfig}
\usepackage{xparse}
\usepackage{xargs,ifthen,xstring,etoolbox}
\usepackage{environ}
\usepackage{multirow}
\usepackage{array}
\usepackage{hyperref}
\usepackage{amsfonts}
\usepackage{dsfont}
\usepackage{bbm}
\usepackage{caption}
\usepackage{subcaption}
\usepackage{xspace}
\usepackage{xcolor,colortbl}
\usepackage{tcolorbox}
\usepackage{authblk}
\usepackage{soul}
\usepackage{csquotes}
\usepackage{tikz-network}
\usepackage{float}

\usepackage{doi}
\usepackage[numbers]{natbib}
\usepackage{tikz}

\usetikzlibrary{patterns,calc,matrix,backgrounds,fadings,shapes,arrows,shadows,decorations.pathreplacing,shadows.blur}
\usetikzlibrary{trees}
\usetikzlibrary{quantikz2}
\usetikzlibrary{tikzmark}
\tikzset{
    gateO/.style={
        draw,
        circle,
        minimum width=1em,
        inner sep=0pt,
        fill=white,
    }
}
\DeclareExpandableDocumentCommand{\pctrl}{m}{%
  \octrl[style={scale=1}]{} \hspace{-0.5em} \adjustbox{width=1em}{#1}
}

\newcommand\restr[2]{{
    \left.\kern-\nulldelimiterspace 
    #1 
    \vphantom{\big|} 
    \right|_{#2} 
}}

\usepackage{amsthm}
\usepackage{letltxmacro}
\LetLtxMacro\amsproof\proof
\LetLtxMacro\amsendproof\endproof
\usepackage{thmbox}
\AtBeginDocument{%
  \LetLtxMacro\proof\amsproof
  \LetLtxMacro\endproof\amsendproof
}
\newtheorem{theorem}{Theorem}
\newtheorem{lemma}{Lemma}
\newtheorem{corollary}{Corollary}
\newtheorem{definition}{Definition}
\newtheorem{remark}{Remark}

\newcommand\nnze[0]{\ensuremath{|B|}}

\newcommand\lognnze[0]{\log(\nnze{})}
\newcommand\Order[1]{\ensuremath{\mathcal{O}\left(#1\right)}}

\newcommand{\RX}[1][\relax]{\ensuremath{R_{X\ifx\relax#1\relax\else,\,#1\fi}}}
\newcommand{\RY}[1][\relax]{\ensuremath{R_{Y\ifx\relax#1\relax\else,\,#1\fi}}}
\newcommand{\RZ}[1][\relax]{\ensuremath{R_{Z\ifx\relax#1\relax\else,\,#1\fi}}}
\newcommand{\RZZ}[1][\relax]{\ensuremath{R_{ZZ\ifx\relax#1\relax\else,\,#1\fi}}}

\newcommand\Id[0]{\ensuremath{\mathbb{I}}}
\newcommand\ZM[0]{\ensuremath{O}}
\newcommand{\maxkcut}[1]{MAX $#1$-CUT}

\newcommand{\Span}[1]{\operatorname{span}\!\left(#1\right)}
\newcommand{\CBU}[0]{C_{B}U}

\newcommand\x[0]{\mathbf{x}}
\newcommand\y[0]{\mathbf{y}}
\newcommand\z[0]{\mathbf{z}}
\newcommand\w[0]{\mathbf{w}}

\renewcommand\a[0]{\mathbf{a}}
\renewcommand\b[0]{\mathbf{b}}
\renewcommand\c[0]{\mathbf{c}}
\renewcommand\t[0]{\mathbf{t}}

\newcommand{\clr}[0]{\operatorname{clr}}

\newcommand{\overbar}[1]{\mkern 1.5mu\overline{\mkern-1.5mu#1\mkern-1.5mu}\mkern 1.5mu}
\newcommand{\lX}[0]{\overbar{X}}

\newcommand\Ph[0]{Ph}

\title{
Compact Circuits for Constrained Quantum Evolutions of Sparse Operators 
}

\author[$\dagger$]{Franz G. Fuchs}
\author[$\dagger$]{Ruben P. Bassa\thanks{\href{mailto:ruben.bassa@sintef.no}{ruben.bassa@sintef.no}}}

\affil[$\dagger$]{SINTEF AS, Department of Mathematics and Cybernetics, Oslo, Norway}

\date{\today}

\begin{document}
\maketitle

\begin{abstract}
We introduce a general framework for constructing compact quantum circuits that implement the real-time evolution of Hamiltonians of the form $H = \sigma P_B$, where $\sigma$ is a Pauli string commuting with a projection operator $P_B$ onto a subspace of the computational basis. Such Hamiltonians frequently arise in quantum algorithms, including constrained mixers in QAOA, fermionic and excitation operators in VQE, and lattice gauge theory applications.
Additionally, we construct transposition gates, widely used in quantum computing, that scale more efficiently than the best known constructions in literature. 
Our method emphasizes the minimization of non-transversal gates, particularly T-gates, critical for fault-tolerant quantum computing. We construct circuits requiring $\mathcal{O}(n|B|)$ CX gates and $\Order{n \nnze{} + \log(\nnze{}) \log (1/\epsilon)}$ T-gates, where $n$ is the number of qubits, $|B|$ the dimension of the projected subspace, and $\epsilon$ the desired approximation precision. For subspaces that are generated by Pauli X-orbits we further reduce complexity to $\mathcal{O}(n \log |B|)$ CX gates and $\Order{n+\log(\frac{1}{\epsilon})}$ T gates.  Our constructive proofs yield explicit algorithms and include several applications, such as improved transposition circuits, efficient implementations of fermionic excitations, and oracle operators for combinatorial optimization.
In the sparse case, i.e. when $\nnze{}$ is small, the proposed algorithms scale favourably when compared to direct Pauli evolution. 
\end{abstract}

\begin{center}
  \captionsetup{hypcap=false}        
    \centering
    \begin{tikzpicture}[edge from parent fork down]
    \tikzset{level 1/.style={sibling distance=9cm, level distance = 1cm}},
    \node { $\sigma=i^{\a \cdot \b} X^\a Z^\b$}
    child [level 2/.style={sibling distance=4.5cm, level distance = 1cm},
           ] {
        node[yshift=.0cm] {$X^\a=I$}
        child [level 3/.style={sibling distance=2.0cm, level distance=1.75cm},
           ] {
        node[inner sep=3pt, yshift=0cm, fill=white] {$Z^\b=I$}
            child[edge from parent path={(\tikzparentnode.south) -- (\tikzchildnode.north)}] {
                node [draw, rounded corners, inner sep=3pt, text width=4.0cm] {
                \begin{center}
                $e^{it H} =$\\
                    \adjustbox{height=1cm}{
                    \begin{quantikz}[row sep={0.75cm,between origins},column sep=.24ex]
                        &\qwbundle{} & \gate[
                        style={draw, shape=rectangle, rounded corners=0,text height=1.3cm, text width=.5cm,text depth=0},
                        2]{M} & 
                        \gate[style={draw,circle,minimum width=.25em,inner sep=-3pt,fill=white}]{\scriptscriptstyle\leq \! K}
                        \arrow[arrows, yshift=0cm]{d}  & \gate[
                        style={draw, shape=rectangle, rounded corners=0,text height=1.3cm, text width=.5cm,text depth=0},
                        2]{M^\dagger} && \qwbundle{} \arrow[arrows]{lllll} \\
                        &\qw & & \gate[style={draw, shape=rectangle, rounded corners=0,text height=.6cm, text width=1cm,text depth=0}]{\Ph(t)} & & \qw \arrow[arrows]{lllll}
                    \end{quantikz}
                    }
                \end{center}
                }
            }
        }
        child [level 3/.style={sibling distance=2.0cm, level distance=1.75cm},
           ]{
        node[inner sep=3pt, yshift=0cm, fill=white] {$Z^\b \neq I$}
            child[edge from parent path={(\tikzparentnode.south) -- (\tikzchildnode.north)}] {
                node [draw, rounded corners, inner sep=3pt,text width=4.0cm] {
                \vspace{-1\baselineskip}
                \begin{center}
                $e^{itH} =\prod_{s\in\{+,-\}}$\\
                    \adjustbox{height=1cm}{
                    \begin{quantikz}[row sep={0.75cm,between origins},column sep=.24ex]
                        &\qwbundle{} & \gate[
                        style={draw, shape=rectangle, rounded corners=0,text height=1.3cm, text width=.5cm,text depth=0},
                        2]{M_{\scriptscriptstyle s}} & 
                        \gate[style={draw,circle,minimum width=.25em,inner sep=-4pt,fill=white}]{\scriptscriptstyle\leq \! K_{\scriptscriptstyle s}}
                        \arrow[arrows, yshift=0cm]{d}  & \gate[
                        style={draw, shape=rectangle, rounded corners=0,text height=1.3cm, text width=.5cm,text depth=0},
                        2]{M^\dagger_{\scriptscriptstyle s}} && \qwbundle{} \arrow[arrows]{lllll} \\
                        &\qw & & \gate[style={draw, shape=rectangle, rounded corners=0,text height=.6cm, text width=1.2cm,text depth=0}]{\Ph(s t)} & & \qw \arrow[arrows]{lllll}
                    \end{quantikz}
                    }
                \end{center}
                }
            }
        }
    }
    child [level 2/.style={sibling distance=4.5cm, level distance = 1cm},
           ] {
        node[yshift=.0cm]  {$X^\a\neq I$}
        child  [level 3/.style={sibling distance=2.0cm, level distance=1.75cm},
           ]{
        node[inner sep=3pt, yshift=0cm, fill=white] {$[X^\a,Z^\b]\neq 0$}
            child[edge from parent path={(\tikzparentnode.south) -- (\tikzchildnode.north)}] {
                node [draw, rounded corners, inner sep=3pt,text width=4.0cm] {
                \vspace{-.75\baselineskip}
                \begin{center}
                $e^{itH} = $\\
                    \adjustbox{height=1cm}{
                    \begin{quantikz}[row sep={0.75cm,between origins},column sep=.24ex]
                        &\qwbundle{} & \gate[
                        style={draw, shape=rectangle, rounded corners=0,text height=1.3cm, text width=.5cm,text depth=0},
                        2]{M} & 
                        \gate[style={draw,circle,minimum width=.25em,inner sep=-3pt,fill=white}]{\scriptscriptstyle\leq \! K}
                        \arrow[arrows, yshift=0cm]{d}  & \gate[
                        style={draw, shape=rectangle, rounded corners=0,text height=1.3cm, text width=.5cm,text depth=0},
                        2]{M^\dagger} && \qwbundle{} \arrow[arrows]{lllll} \\
                        &\qw & & \gate[style={draw, shape=rectangle, rounded corners=0,text height=.6cm, text width=1cm,text depth=0}]{\RY(t)} & & \qw \arrow[arrows]{lllll}
                    \end{quantikz}
                    }
                \end{center}
                }
            }
        }
        child [level 3/.style={sibling distance=2.0cm, level distance=1.75cm},
           ]{
        node[inner sep=3pt, yshift=0cm, fill=white] {$[X^\a,Z^\b] = 0$}
            child[edge from parent path={(\tikzparentnode.south) -- (\tikzchildnode.north)}] {
                node [draw, rounded corners, inner sep=3pt,text width=4.0cm] {
                \vspace{-1\baselineskip}
                \begin{center}
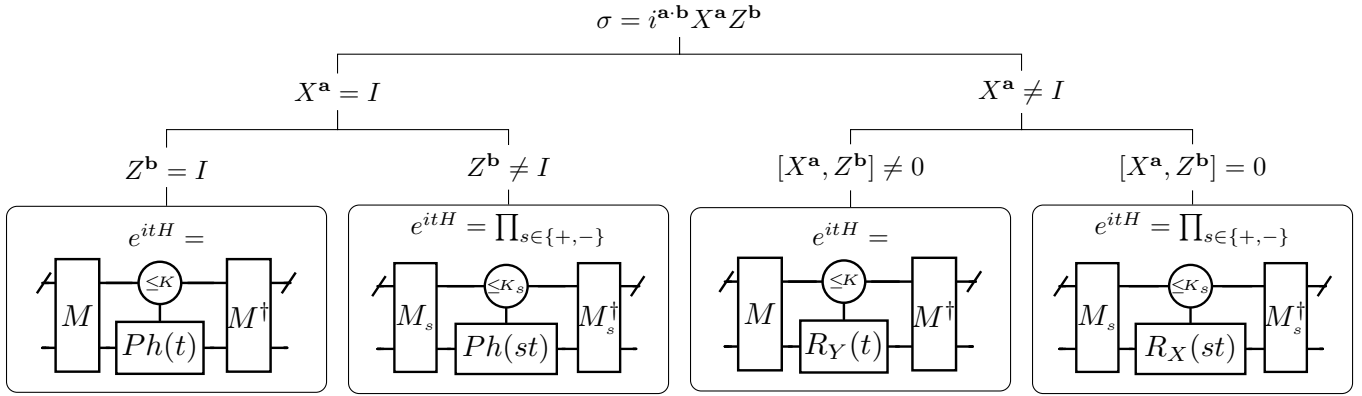

                $ e^{itH} =\prod_{s\in\{+,-\}}$ \\
                    \adjustbox{height=1cm}{
                    \begin{quantikz}[row sep={0.75cm,between origins},column sep=.24ex]
                        &\qwbundle{} & \gate[
                        style={draw, shape=rectangle, rounded corners=0,text height=1.3cm, text width=.5cm,text depth=0},
                        2]{M_{\scriptscriptstyle s}} & 
                        \gate[style={draw,circle,minimum width=.25em,inner sep=-4pt,fill=white}]{\scriptscriptstyle\leq \! K_{\scriptscriptstyle s}}
                        \arrow[arrows, yshift=0cm]{d}  & \gate[
                        style={draw, shape=rectangle, rounded corners=0,text height=1.3cm, text width=.5cm,text depth=0},
                        2]{M^\dagger_{\scriptscriptstyle s}} && \qwbundle{} \arrow[arrows]{lllll} \\
                        &\qw & & \gate[style={draw, shape=rectangle, rounded corners=0,text height=.6cm, text width=1.3cm,text depth=0}]{\RX(s t)} & & \qw \arrow[arrows]{lllll}
                    \end{quantikz}
                    }
                \end{center}
                }
            }
        }
    }
;
\end{tikzpicture}
    \captionof{figure}{
    Given a Pauli string $\sigma$ and a commuting projection operator $P_B$, the real time evolution of $H=\sigma P_B$ can be realized efficiently for small $|B|$ with permutation operators $M$
    described in Section~\ref{sec:generalcase}
    and low-pass controlled unitary gates with $K_{(s)}\leq |B|$, which are introduced in Section~\ref{sec:lowpass}.
    The specific unitary can be classified into four distinct cases depending on the relation of the $X$ and $Z$ terms in $\sigma$. 
    Here, $\Ph(t) = \ket{0}\bra{0} + e^{i t} \ket{1}\bra{1}$ is the phase shift gate and $\RX$, $\RY$ is the rotational $X$, $Y$ gate.
    }
    \label{fig:overview}
    \captionsetup{hypcap=true}
\end{center}

\section{Introduction and related work}
In this article we propose a general method to construct quantum circuits that minimize the number of non-transversal gates for realizing the 
real time evolution of a Hamiltonian of the form $H = \sigma P_B$,
where $\sigma$ is a Pauli string that commutes with the projection operator $P_B$, onto a subspace of the Hilbert space.
Evolution operators $e^{-itH}$ of this form show up in several important applications.
In VQE~\cite{peruzzo2014variational} fermionic excitation operators in second quantization take this form after applying the Jordan-Wigner transform. Similarly, in QAOA~\cite{farhi2014quantum,hadfield2019quantum}, both phase separating operators~\cite{fuchs2021efficient} and constrained mixers acting on subspaces~\cite{fuchs2023LX} can exhibit this structure.
Other examples are the trace gate for lattice gauge theory~\cite{Alam_2022}, and the transposition gate that permutes two computational basis state, which is widely found in quantum computing.

A generic way of approximating the real time evolution of any Hamiltonian on a gate based quantum computer is to decompose it in the Pauli basis
and then realize the evolution of each term of the weighted sum through a circuit given in Figure~\ref{fig:Pauliexp}.
This is exact if the terms commute, otherwise we have to employ a Trotterization.
There are two arguments against using this construction to create a circuit.
Firstly, the decomposition requires in general to evaluate $4^n$ Hilbert-Schmidt inner products,
since one needs to expressing $P$ as a weighted sum of Pauli strings.
Secondly, evaluating the circuit depth on NISQ device, which is dominated by CX-gate count, is vastly different to the fault tolerant setting, where also rotational gates can have a dominant footprint.

Fault tolerant quantum computers typically use a discrete set of gates, with one of the most common choices being Clifford$+$T gates.
A $ 2 \times 2 $ unitary matrix can be exactly expressed using Clifford+T gates if and only if its entries belong to the ring $ D[\omega]$~\cite{kliuchnikov2013exact}, where $ D[\omega]$ is the set of dyadic unitary numbers generated by $ \omega = e^{i\pi/4}$.
If an operation such as $ R_z(t) $ has matrix entries outside this ring, it must be approximated by a sequence of Clifford $+$ T gates.
This result is fundamental to quantum compiling techniques such as the Gridsynth algorithm~\cite{ross2016optimal} and other Solovay-Kitaev-like methods, which approximate arbitrary single-qubit rotations using Clifford+T circuits with a given precision $\epsilon$.
As a typical example, the approximation of $R_z(\pi/128)$ up to $\varepsilon = 10^{-10}$  
by ancilla-free Clifford+T circuits requires at least 102 T-gates~\cite{ross2016optimal}.
In the typical case, ancilla-free circuit approximations require
$\log(1/\varepsilon)$
T-gates.
On a fault-tolerant quantum device using the surface code performing a logical CX gate is much faster ($\mu s$) than magic state distillation ($ms$) in most practical scenarios~\cite{fowler2013surface,litinski2019magic}.

This motivates to derive a method to construct compact circuits for realizing $e^{-i t \sigma P_B}$ described above.
Here, compact means using as few rotational gates as possible.
Related work include circuit constructions for Hamiltonian simulation~\cite{yordanov2020efficient} and phase gadget synthesis~\cite{Cowtan2019phasegadget}.
One can also use the QR factorization for gate decomposition in certain settings~\cite{mottonen12006decompositions}.
Another central tool for unitary synthesis are recursive Cartan decompositions~\cite{Wierichs2025}, which provide a way to exactly factorize quantum circuits into smaller components.
The fundamental routine of state preparation is a closely related topic, for which many algorithms have been proposed.
There exist methods that have linear gate and qubit complexity in the number of non-zero amplitudes~\cite{ramacciotti2024simple}.

Since our algorithm is expressed in terms of MCX, transposition, and MCU gates it is useful to know their resource costs in terms of Cliffort+T gates, as summarized in
Table~\ref{tab:resource_estimates}.
In particular, the \emph{multi-controlled NOT (MCX)} gate, also called the $n$-Toffoli gate, is defined as  
\begin{equation*}
    C^n X= \ket{1}\bra{1}^{\otimes n}\otimes X+(\Id-\ket{1}\bra{1}^{\otimes n})\otimes \Id.
\end{equation*}
Different implementations optimize circuit complexity with at least one ancilla. A linear-depth and size approach ($\Order{n}$) is given in \cite{Decompositionsofn-qubitToffoliGateswithLinearCircuitComplexity}, while a recursive decomposition in \cite{Claudon_2024} achieves $\Order{\log(n)^3}$ depth at the cost of $\Order{n \log(n)^4}$ size and a higher T-count.
A \emph{transposition gate} swaps two computational basis states:
\begin{equation*}
    T_{\x,\y} = \ket{\x}\bra{\y} + \ket{\y}\bra{\x} + \sum_{\w\neq \x,\y} \ket{\w}\bra{\w}.
\end{equation*}
A near-optimal construction requires $\Theta(n)$ $X, CX$ gates, two $C^{n+1}X$ gates, and one clean ancilla \cite{herbert2024almost}. Our method achieves $T_{\x,\y}$ without ancillas, using $\Order{n}$ $X, CX$ gates, $\Order{\log(n)}$ depth, and one $C^nX$, based on Theorem~\ref{theorem:GroupGeneratedCase}.
\begin{figure}
\centering
    \begin{minipage}[t]{0.57\textwidth}
        \centering
        \adjustbox{scale=0.9}{
        \begin{tabular}{lllr}
            \toprule
             gate & \#CX & \#T=depth& \#anc \\
             \midrule
             $\RX,\RY,\RZ$             & 0 & $\Order{\log(\frac{1}{\epsilon})}$ & 0 \\
             $C^n X, T_{x,y}$             & $\Order{n}$ & $\Order{n}$ & 1 \\
             $C^n U, U\in SU(2)$ & $\Order{n}$ & $\Order{n+\log(1/\epsilon)}$ & 0 \\
             $C^{n-k} e^{it\sigma}, \sigma \in S_k$ & $\Order{n}$ & $\Order{n+\log(1/\epsilon)}$ & 0 \\
             $C_{\leq K} U, U\in SU(2)$ & $\Order{n\log(K)}$ & $\Order{\left(n+ \log (1/\epsilon)\right)\log(K)}$ & 0 \\
             \bottomrule
        \end{tabular}
        }
        \captionof{table}{
        Asymptotic resource requirements for common gates in terms of \#CX and $T$ gates, as well as number of ancilla qubits (\#anc).
        The parameter $\epsilon$ denotes the target approximation precision.
        $T_{x,y}$ is the transposition gate between two computational basis states$\ket{\x}$ and $\ket{y}$,
        $C^nX$ is the $n$-controlled NOT gate, $C^n U$ is a multi-controlled unitary, and 
        $S_k$ is the Pauli group on $k$ qubits.
        In addition, we include the scaling of the low-pass controlled gate introduced in Section~\ref{sec:lowpass}.
        }
        \label{tab:resource_estimates}
    \end{minipage}
    \hfill
    \begin{minipage}[t]{0.41\textwidth}
        \centering
        $C_\mathbf{1}^{\c}e^{it\sigma}_{\t} \! \! \! =$
        \hspace{-.8cm}
        \begin{tikzpicture}[remember picture, baseline={0cm}]
    \node[scale=.75] {
    \begin{quantikz}[row sep={0.5cm,between origins},column sep=1ex, wire types = {q,q,q,n,q,q,q}]
    &\gate[wires=6]{U}
    \lstick{$\t_k$}& \ctrl{1} &  &  & & &  &  &  & & & & \ctrl{1} & \gate[wires=6]{U^\dagger}& \\
    && \targ{} &\ctrl{1} &  & & &  &  &  & & & \ctrl{1} &\targ{}&& \\
    & &   & \targ{} &  &  & & &   &  & & &\targ{} &  &  & \\
    && &&\raisebox{2mm}{$\ddots$}&&&&&\raisebox{2mm}{$\adots$}&&&&&&\\
    &&  &  &  &  &  \ctrl{1} & & \ctrl{1} & &   &  &  &  & &  \\
    \lstick{$\t_1$}&&  &    &  &  & \targ{} & \gate{R_z(2 t)} & \targ{}&  &  &  &  &  &  &  \\
    \lstick{$\c$}&\qwbundle[]{}& &&&&&\ctrl{-1}&&&&&&&&\qwbundle[]{}
    \end{quantikz}
    };
\end{tikzpicture}

        \captionof{figure}{
        Standard circuit to realize the MCU gate for $U = e^{-i t \sigma}$, with control state $\mathbf{1} = 1\cdots 1$ on control register $\c=(1, \cdots, n-k)$ and target register $\t=(n-k+1, \cdots, n)$.
        The unitary $U = \bigotimes_{i=1}^k U_i$ realizes the basis change $Z = U_i\sigma_iU_i^\dag $.
        Note that the same circuit holds for the uncontrolled version.
        }
        \label{fig:Pauliexp}
    
    \end{minipage}%
\end{figure}
A \emph{multi-controlled unitary (MCU)} gate is given by
\begin{equation}
    \label{eq:MCU}
    C^\c_\b U_\t =
    (I_\c - \ket{\b}\bra{\b}_\c) \otimes I_\t
    + \ket{\b}\bra{\b}_\c \otimes U_\t,
\end{equation}
where $\mathbf{c}, \mathbf{t}$ index control and target qubits, and $\mathbf{b}$ specifies the control condition.
\begin{itemize}
    \item For $U \in SU(2)$ and $n$ control qubits, a decomposition with $\Order{n}$ depth and $CX$ gates exists \cite{vale2023decompositionmulticontrolledspecialunitary}. The T-count scales as $\Order{n+\log(1/\epsilon)}$.
    \item For $U=e^{i t \sigma},\sigma \in S_k$, i.e. a length $k$ Pauli string and $n-k$ control qubits, a decomposition with $\Order{n}$ depth and CX gates and $\Order{n+\log(1/\epsilon)}$ T-gates exists.
    The circuit for this case is shown in Figure~\ref{fig:Pauliexp}. Let $U_\sigma$ be the basis change and $M$ the CNOT stairs from the Figure,
    we have that $\sigma=U_\sigma^{\dag}M\mathbb{I}\otimes ZM^{\dag}U_\sigma$ and consequently $U=e^{i t \sigma}=U_\sigma^{\dag}M^{\dag}C_n\RZ(t)MU_\sigma$.
\end{itemize}

The main contribution is presented in Section~\ref{sec:main} providing an efficient construction of the real time evolution of Hamiltonians of the form $H=\sigma P_B$ for $[\sigma,P_B]=0$ in the sparse case, i.e., if $\nnze{}$ is small.
The proof for the special case of Pauli X-orbit generated subspaces is provided in Section~\ref{sec:GroupGenerated},
and for the general case in Section~\ref{sec:generalcase}.
To be able to proof the general case we introduce the concept of subspace-controlled unitary gates in Section~\ref{sec:sscug}.
In particular, we introduce a gate in Section~\ref{sec:lowpass} that applies a unitary conditioned on the reference state being in one of the first/last $K$ computational basis states.
We show that this gate, which we dub low/high-pass controlled unitary gate, admits an efficient realization.
In Section~\ref{sec:examples} we provide several examples, where our method can be applied.
One example is the transposition gate, where, to the best of our knowledge, the most efficient explicitly constructed algorithm for transposing any computational basis state requires $\Order{n}$ X and CX gates, two MCX gates, and one clean ancilla qubit~\cite{herbert2024almost},
whereas we achieve a realization with $\Order{n}$ X, CX gates, only one MCX gate  (with one less control), and no ancilla qubit.

\section{Subspace-controlled unitary gates}
\label{sec:sscug}

Central to this paper is the concept of subspace-controlled unitary operations, which can be understood as a generalization of multi-controlled unitary gates.
We start by defining the general concept, before we a introduce unitary gate, controlled by the first or last $k$ computational basis states, which serves as a building block for our main theorem.

\subsection{Theory}

In this paper a (sub)set of \emph{computational basis states} of the Hilbert space of $n$ qubits is denoted by
\begin{equation}
    \label{eq:B}
    B=\left\{\ket*{\z_j}| \ \ 1\leq j\leq J, \ \ \z_j\in\{0,1\}^n, \ \ \z_j\neq \z_j \text{ for } i\neq j \right\},
\end{equation}
and the \emph{projector} onto the subspace spanned by a $B$ is given by
\begin{equation}
    \label{eq:P_B}
    P_B \coloneqq \sum_{\ket{\z}\in B}\ket{\z}\bra{\z}.
\end{equation}
This allows us to generalize the notion of controlled gates through the following definition.
\begin{definition}[Subspace-controlled Unitary]
    Given a unitary $U:\mathcal{H}\rightarrow\mathcal{H}$ and a projector $P_B:\mathcal{H}\rightarrow\mathcal{H}$ with $[U,P_B] = 0$, we define
    a unitary operator controlled by the subspace $B$ as
    \begin{equation*}
        \CBU \coloneqq (I-P_{B}) + P_{B}UP_{B},
    \end{equation*}
    i.e. it acts as a unitary $U$ on the subspace $\Span{B}$ and as the identity on $\Span{B}^\perp$.
\end{definition}
We remark that this is well defined;
Assuming $[U ,P_B] = 0$, we compute the product of $\CBU$ with its adjoint.
Using that $U$ is unitary, $P_B$ is a projection and the commutation property it follows directly that
\begin{equation*}
    \CBU(\CBU)^{\dag}
    =I-P_{B}
    +P_{B}UP_B U^{\dagger}P_{B}
 = I,
\end{equation*}
and similarly for $(\CBU)^\dagger \CBU$ follows.
So $\CBU$ is indeed a unitary operator.
It is also easy to show the other direction that if $\CBU$ is unitary then $[P_B,U]$ = 0.
Note that the definition is consistent with multi-controlled unitary gates given in Equation~\eqref{eq:MCU}, since
\begin{equation*}
\begin{split}
    C_{ \{ \ket{\b}_\c \otimes I_\t \} } (I_\c \otimes U_\t) &= 
    \big( I-\ket{\b}\bra{\b}_\c \otimes I_\t \big) + 
    \big( \ket{\b}\bra{\b}_\c \otimes I_\t  \big)  \big( I_\c \otimes U_\t \big)   \big(   \ket{\b}\bra{\b}_\c \otimes I_\t \big)
    \\
    &=
    (I_\c-\ket{\b}\bra{\b}_\c) \otimes I_\t + 
    \ket{\b}\bra{\b}_\c \otimes U_\t
    =
    C^\c_\b U_\t.
\end{split}
\end{equation*}
A transposition gate can also be interpreted as a subspace controlled unitary
\begin{equation*}
    T_{x,y} = C_{\{\ket{x},\ket{y}\}} X^{x \oplus y},
\end{equation*}
where $x \oplus y$ is component-wise addition modulo 2.

As a reminder, the set of all \emph{Pauli strings} of length $n$ is given by
\begin{equation}
    \label{eq:PauliStrings}
    S_n \coloneqq \left\{ i^{\a \cdot \b} X^\a Z^\b \mid \ \ \a, \b \in \mathbb{Z}_2^n, \ \ 
\a \cdot \b = \sum\nolimits_{j=1}^{n} a_j b_j \!\!\!\! \pmod{4}
    \right\},
\end{equation}
where  
$
X^\a = X^{a_1} \otimes \dots \otimes X^{a_n}, \quad
Z^\b = Z^{b_1} \otimes \dots \otimes Z^{b_n}.
$
We can now interpret the time evolution of $H=\sigma P_B$ as a subspace-controlled unitary if $\sigma$ and $P_B$ commute.
\newpage
\begin{lemma}[Time evolution is subspace-controlled rotation]
    Let $P_B$ be a projector onto a subspace $B$ and $\sigma \in S_n$ be a Pauli string with $[\sigma,P_B] = 0$.
    Then for the real time evolution of $H=\sigma P_B$ we have that
    \begin{equation*}
        e^{i t \sigma P_{B}} = C_{B}e^{i t \sigma},
    \end{equation*}
    i.e. it is a subspace controlled Pauli evolution.
\end{lemma}
\begin{proof}
From the commutation relations and the properties of a projector and Pauli strings we obtain
\begin{equation}
\label{eq:Hpower}
    (\sigma P_{B})^{2j} = P_{B}, \quad
    (\sigma P_{B})^{2j+1} = \sigma P_{B}.
\end{equation}
Expanding the matrix exponential of $H$ and using the results above, we get:
\begin{equation*}
\begin{split}
    e^{i t \sigma P_{B}} &=
    \sum_{j=0}^{\infty}\frac{(i t)^{j}}{j!}(\sigma P_{B})^{j}=
    I + \sum_{j=1}^{\infty}\frac{(i t)^{2j}}{(2j)!}(\sigma P_{B})^{2j}
    +
        \sum_{j=0}^{\infty}\frac{(i t)^{2j+1}}{(2j+1)!}(\sigma P_{B})^{2j+1}\\
    &\underset{\eqref{eq:Hpower}}{=}
    I+(\cos(t)-1)P_{B} + i\sin(t) \sigma P_{B} =
    (I-P_B) + (\cos(t) + i\sin(t) \sigma) P_{B} =
    C_Be^{it\sigma},
\end{split}
\end{equation*}
showing the assertion.
\end{proof}

\subsection{High- and low-pass controlled unitary gates}
\label{sec:lowpass}
Essential for our construction is the efficient realization of a gate that applies a unitary conditioned on the reference state being in one of the first $K$ computational basis states, which we dub low-pass controlled unitary gate.
It is defined as follows
\begin{equation*}
    C_{\leq K} U \coloneqq
    \sum_{i<K}\ket{i}\bra{i}_\c\otimes U_t+\sum_{i\geq K}\ket{i}\bra{i}_\c\otimes\Id_t=
    (\oplus_K U) (\oplus_{2^n-K} I) = 
    \adjustbox{scale=0.75}{
    $
    \begin{bmatrix}
        U      & \ZM    &        & \cdots  &        & \ZM \\
        \ZM    & \ddots &        &         &        &     \\
               &        & U      & \ddots  &        & \vdots \\
        \vdots &        & \ddots & I       &        &     \\
               &        &        &         & \ddots & \ZM \\
        \ZM    &        & \cdots &         & \ZM    & I
    \end{bmatrix}
    $
    },
\end{equation*}
where $\oplus$ denotes the direct sum of matrices.

\begin{theorem}[Low-pass controlled unitary gate]
\label{theorem:firstK}
    The low-pass controlled unitary gate can be expressed as a product of $\Order{\log(K)}$ multi-controlled unitaries, i.e.
    \begin{equation*}
        C_{\leq K}^\c U_\t=\prod_{i=0}^{p-1}C_{c_i(K)}^{\c}U_{\t},
    \end{equation*}
    where $c_i(K)$ is a bitstring given in Equation~\eqref{eq:cond_states}.
\end{theorem}

\begin{proof}
    First, we express $K$ in its binary representation
    \begin{equation*}
        K = 2^{k_1}+2 ^{k_2} \cdots+ 2^{k_p}, \text{ s.t. }k_1>k_2>\cdots>k_p,
    \end{equation*}
    which means that $k_i$ are the indices of the binary strings that are 1 in big-endian byte encoding.
    We partition the integer interval $[0,K-1]$ into successive segments with dimension that are powers of two determined by the presence of $1$s in the binary expansion of $K$.
    Given a number $1\leq K \leq 2^n$ we define an auxiliary sequence that represents the cumulative sum of powers of two from $0$ to $i$ given by 
    \begin{equation*}
        K_j = 
        \sum_{i=1}^{j} 2^{k_j}, \quad 0\leq j\leq p.
    \end{equation*}
    It is easy to check that the binary representation of $K_j$ to $K_{j+1}-1$ has the first $n-k_j$ entries fixed, while the remaining indices go through all possible combinations of bitstrings from $0$ to $K_{j+1}-1$.
    The bits in common for the interval $[K_j, K_{j+1}-1]$ are given by the binary representation of
    \begin{equation}
        \label{eq:cond_states}
        c_j(K) = \sum_{i=1}^{j-1} 2^{k_i - k_j},
    \end{equation}
    which allows us to express
    \begin{equation*}
        \sum_{i=K_j}^{K_{j+1}-1}\ket{i}\bra{i} = \ket{c_j(K)}\bra{c_j(K)}\otimes \Id^{\otimes k_j}
    \end{equation*}
    
    Using this notation, the first $K$-state controlled unitary operator can be decomposed as follows:
    \begin{equation*}
        \begin{split}
            C_{\leq K}^\c U_\t
            &=
            \sum_{i<K}\ket{i}\bra{i}_\c\otimes U_t+\sum_{i\geq K}\ket{i}\bra{i}_\c\otimes\Id_t \\
            &=
            \prod_{j=0}^{p-1} \left( \sum_{i=K_j}^{K_{j+1}-1} \ket{i}\bra{i}_\c \otimes U_\t + \sum_{i\notin (K_j,K_{j+1}]} \ket{i}\bra{i}_\c \otimes \Id_\t \right) \\
            &=
            \prod_{j=0}^{p-1} \left( \sum_{i=K_j}^{K_{j+1}-1} (\ket{c_j(K)}\bra{c_j(K)}\otimes \Id^{\otimes k_j})_\c \otimes U_\t + \sum_{i\notin (K_j,K_{j+1}]} \ket{i}\bra{i}_\c \otimes \Id_\t \right) \\
            &=
            \prod_{j=0}^{p-1}C_{c_j(K)}^{\c}U_{\t}.
        \end{split}
    \end{equation*}
    where each term in the product can be interpreted as a block matrix acting on the subspace spanned by basis states between $K_j$ and $K_{j+1}$, while the identity operation applies elsewhere.
\end{proof}

\begin{remark}
    In the special case where $U \in SU(2)$, the 
    number of CX gates scales as $\Order{n\log(K)}$ and the depth and number of T gates
    scales as $\Order{\left(n+ \log (\frac{1}{\epsilon})\right)\log(K)}$. 
    The second term related to $\epsilon$ comes from the overhead from approximating single-qubit rotational gates.
\end{remark}

As an example $K = 42 = 2^5 + 2^3 + 2^1$, i.e, $\mathbf{k} = (5,3,1)$.
The intervals are given by
\begin{equation*}
    \begin{split}
        [K_0, K_1 - 1] &= [\underline{0}00000, \underline{0}11111],\\
        [K_1, K_2 - 1] &= [\underline{100}000, \underline{100}111],\\
        [K_2, K_3 - 1] &= [\underline{10100}0, \underline{10100}1],
    \end{split}
\end{equation*}
where the common bits, $c_1(K)=0, c_2(K)=100$, and $c_3(K)=10100$, in the ranges are underlined.
Overall, this leads to the circuit
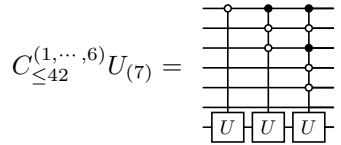
\begin{equation*}
    C^{(1,\cdots,6)}_{\leq 42}U_{(7)}=
    \begin{tikzpicture}[remember picture, baseline={0cm}]
        \node[scale=.75] {
        \begin{quantikz}[row sep={0.35cm,between origins},column sep=1ex]
        \qw & \octrl{6} & \ctrl{1}  & \ctrl{1}  & \qw  \\
        \qw &           & \octrl{1} & \octrl{1} & \qw  \\
        \qw &           & \octrl{4} & \ctrl{1}  & \qw  \\
        \qw &           &           & \octrl{1} & \qw  \\
        \qw &           &           & \octrl{2} & \qw  \\
        \qw &           &           &           & \qw  \\
        \qw & \gate{U}  & \gate{U}  & \gate{U}  & \qw   
        \end{quantikz}.
        };
    \end{tikzpicture}
\end{equation*}
In conclusion, the circuit consists of a sequential application of multi-controlled unitaries, where the number of control qubits increases according to the binary representation of K. Since the cost of $C^{n}U$ gate scales linearly in both depth and CX size with respect to the number of controls $\Order{n}$. Summing over all required control levels up to $\log(K)$ the total cost scales as:
\begin{figure}
    \centering
    $C_{\leq 1} U =$
\begin{tikzpicture}[remember picture, baseline={0cm}]
    \node[scale=.5] {
    \begin{quantikz}[row sep={0.5cm,between origins},column sep=1ex]
        \qw & \octrl[style={scale=1.5}]{1}  & \\
        \qw &  \octrl[style={scale=1.5}]{1} & \\
        \qw &  \octrl[style={scale=1.5}]{1} & \\
        \qw & \gate{U} &
    \end{quantikz}
    };
\end{tikzpicture}
    $=$
\begin{tikzpicture}[remember picture, baseline={0cm}]
    \node[scale=.5] {
    \begin{quantikz}[row sep={0.5cm,between origins},column sep=1ex]
        \qw &          & \ctrl[style={scale=1.5}]{3}      &  \octrl[style={scale=1.5}]{1}     &  \octrl[style={scale=1.5}]{1}     & \\
        \qw &          &               &  \ctrl[style={scale=1.5}]{2}      &  \octrl[style={scale=1.5}]{1}     & \\
        \qw &          &               &                &  \ctrl[style={scale=1.5}]{1}      & \\
        \qw & \gate{U} & \gate{U^\dag} &  \gate{U^\dag} &  \gate{U^\dag} & 
    \end{quantikz}
    };
\end{tikzpicture},
$C_{\leq 2} U =$
\begin{tikzpicture}[remember picture, baseline={0cm}]
    \node[scale=.5] {
    \begin{quantikz}[row sep={0.5cm,between origins},column sep=1ex]
        \qw & \octrl[style={scale=1.5}]{1}  & \\
        \qw & \octrl[style={scale=1.5}]{2}       & \\
        \qw &                 & \\
        \qw & \gate{U} & 
    \end{quantikz}
    };
\end{tikzpicture}
    $=$
\begin{tikzpicture}[remember picture, baseline={0cm}]
    \node[scale=.5] {
    \begin{quantikz}[row sep={0.5cm,between origins},column sep=1ex]
        \qw &          & \ctrl[style={scale=1.5}]{3}      & \octrl[style={scale=1.5}]{1}     & \\
        \qw &          &               & \ctrl[style={scale=1.5}]{2}      & \\
        \qw &          &               &               & \\
        \qw & \gate{U} & \gate{U^\dag} & \gate{U^\dag} & 
    \end{quantikz}
    };
\end{tikzpicture},
$C_{\leq 3} U =$
\begin{tikzpicture}[remember picture, baseline={0cm}]
    \node[scale=.5] {
    \begin{quantikz}[row sep={0.5cm,between origins},column sep=1ex]
        \qw & \octrl[style={scale=1.5}]{1}  & \octrl[style={scale=1.5}]{1} & \qw   \\
        \qw & \octrl[style={scale=1.5}]{2}       & \ctrl[style={scale=1.5}]{1}  & \qw \\
        \qw &                 & \octrl[style={scale=1.5}]{1} & \qw  \\
        \qw & \gate{U}        & \gate{U}  & \qw 
    \end{quantikz}
    };
\end{tikzpicture}
    $=$
\begin{tikzpicture}[remember picture, baseline={0cm}]
    \node[scale=.5] {
    \begin{quantikz}[row sep={0.5cm,between origins},column sep=1ex]
        \qw &          & \ctrl[style={scale=1.5}]{3}      & \octrl[style={scale=1.5}]{1}     & \\
        \qw &          &               & \ctrl[style={scale=1.5}]{1}      & \\
        \qw &          &               & \ctrl[style={scale=1.5}]{1}      & \\
        \qw & \gate{U} & \gate{U^\dag} & \gate{U^\dag} & 
    \end{quantikz}
    };
\end{tikzpicture},
$C_{\leq 4} U =$
\begin{tikzpicture}[remember picture, baseline={0cm}]
    \node[scale=.5] {
    \begin{quantikz}[row sep={0.5cm,between origins},column sep=1ex]
        \qw & \octrl[style={scale=1.5}]{3}  & \\
        \qw &                 & \\
        \qw &                 & \\
        \qw & \gate{U} & 
    \end{quantikz}
    };
\end{tikzpicture}
    $=$
\begin{tikzpicture}[remember picture, baseline={0cm}]
    \node[scale=.5] {
    \begin{quantikz}[row sep={0.5cm,between origins},column sep=1ex]
        \qw &          & \ctrl[style={scale=1.5}]{3} & \\
        \qw &          &   & \\
        \qw &          &   & \\
        \qw & \gate{U} & \gate{U^\dag} & 
    \end{quantikz}
    };
\end{tikzpicture},

$C_{\leq 5} U =$
\begin{tikzpicture}[remember picture, baseline={0cm}]
    \node[scale=.5] {
    \begin{quantikz}[row sep={0.5cm,between origins},column sep=1ex]
        \qw & \octrl[style={scale=1.5}]{3} & \ctrl[style={scale=1.5}]{1}  & \\
        \qw &           & \octrl[style={scale=1.5}]{1} & \\
        \qw &           & \octrl[style={scale=1.5}]{1} & \\ 
        \qw & \gate{U}  & \gate{U}  &
    \end{quantikz}
    };
\end{tikzpicture}
    $=$
\begin{tikzpicture}[remember picture, baseline={0cm}]
    \node[scale=.5] {
    \begin{quantikz}[row sep={0.5cm,between origins},column sep=1ex]
        \qw &          & \ctrl[style={scale=1.5}]{1}      & \ctrl[style={scale=1.5}]{1}      & \\
        \qw &          & \ctrl[style={scale=1.5}]{2}      & \octrl[style={scale=1.5}]{1}     & \\
        \qw &          &               & \ctrl[style={scale=1.5}]{1}      & \\
        \qw & \gate{U} & \gate{U^\dag} & \gate{U^\dag} &
    \end{quantikz}
    };
\end{tikzpicture},
$C_{\leq 6} U =$
\begin{tikzpicture}[remember picture, baseline={0cm}]
    \node[scale=.5] {
    \begin{quantikz}[row sep={0.5cm,between origins},column sep=1ex]
        \qw & \octrl[style={scale=1.5}]{3} & \ctrl[style={scale=1.5}]{1}  & \\
        \qw &           & \octrl[style={scale=1.5}]{2} & \\
        \qw &           &           & \\ 
        \qw & \gate{U}  & \gate{U}  &
    \end{quantikz}
    };
\end{tikzpicture}
    $=$
\begin{tikzpicture}[remember picture, baseline={0cm}]
    \node[scale=.5] {
    \begin{quantikz}[row sep={0.5cm,between origins},column sep=1ex]
        \qw &          & \ctrl[style={scale=1.5}]{1}      & \\
        \qw &          & \ctrl[style={scale=1.5}]{2}      & \\
        \qw &          &               & \\
        \qw & \gate{U} & \gate{U^\dag} &
    \end{quantikz}
    };
\end{tikzpicture},
$C_{\leq 7} U =$
\begin{tikzpicture}[remember picture, baseline={0cm}]
    \node[scale=.5] {
    \begin{quantikz}[row sep={0.5cm,between origins},column sep=1ex]
        \qw & \octrl[style={scale=1.5}]{3}  & \ctrl[style={scale=1.5}]{1}                & \qw   \\
        \qw &                 &  \octrl[style={scale=1.5}]{2}    & \qw \\
        \qw &                 &                 & \qw  \\
        \qw & \gate{U} & \gate{U}  & \qw 
    \end{quantikz}
    };
\end{tikzpicture}
    $=$
\begin{tikzpicture}[remember picture, baseline={0cm}]
    \node[scale=.5] {
    \begin{quantikz}[row sep={0.5cm,between origins},column sep=1ex]
        \qw &          & \ctrl[style={scale=1.5}]{1}      & \\
        \qw &          & \ctrl[style={scale=1.5}]{1}      & \\
        \qw &          & \ctrl[style={scale=1.5}]{1}      & \\
        \qw & \gate{U} & \gate{U^\dag} & 
    \end{quantikz}
    };
\end{tikzpicture},
$C_{\leq 8} U =$
\begin{tikzpicture}[remember picture, baseline={0cm}]
    \node[scale=.5] {
    \begin{quantikz}[row sep={0.5cm,between origins},column sep=1ex]
        \qw &  & \\ 
        \qw &  & \\ 
        \qw &  & \\ 
        \qw & \gate{U} &
    \end{quantikz}
    };
\end{tikzpicture},
    \caption{
        All possible low-pass controlled gates $C^\c_{\leq K}U_\t$ for $n=3$ with control register $\c=(1,2,3)$ and target register $\t=(4)$.
        Note, that they can also be realized with specific high-pass controlled gates, namely $(\Id^\c\otimes U_\t) (C^\c_{\geq K+1}U^\dagger_t)$.
    }
    \label{fig:lessthank_n3}
\end{figure}
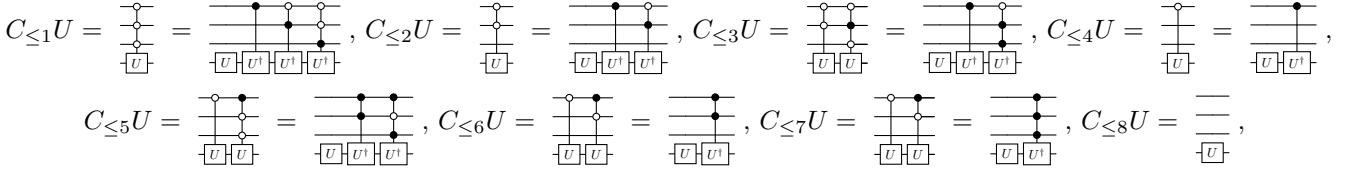
 \begin{equation*}
     \sum_i^{\log(K)}\Order{n-i} = \Order{n\log(K)}.
 \end{equation*}
In the specific case of $U\in SU(2)$ the T count has the same scaling plus a contribution coming from the single qubit approximation and since we have $\log(K)$ multicontrol U this overhead scales as $\Order{\log(k)\log(\frac{1}{\epsilon})}$.
\begin{corollary}[High-pass controlled unitary gate]
    We define the conditional gate that applies a unitary $U$ to a target register $t$ if the control register $c$ is in one of the \emph{last} $K$ computational basis states as
    \begin{equation*}
        C_{\geq K}^\c U_\t=\sum_{i\geq K}\ket{i}\bra{i}_\c\otimes U_t+\sum_{i< K}\ket{i}\bra{i}_\c\otimes\Id_t=\prod_{i=0}^{p-1}C_{c_i(2^n-K+1)+\mathds{1}}^{\c}U_{\t},
    \end{equation*}
    where $c_i(.)$ is a bitstring given in Equation~\eqref{eq:cond_states}.
    The depth and size scale the same as the low-pass controlled gate.
\end{corollary}
The result follow from Theorem~\ref{theorem:firstK} by observing that counting from $0$ to $K-1$ in binary is equivalent to counting from $2^n$ down to $(2^n - K - 1)$ and adding $11\ldots1$ to the strings modulo 2.

Note, that one can always write $C^\c_{\leq K}U_\t = (\Id^\c\otimes U_\t) (C^\c_{\geq K+1}U^\dagger_t)$, which sometimes allows for a more efficient realization.
Figure~\ref{fig:lessthank_n3} shows an example for $n=3$, where we can see that for $K\in{3,6,7}$ this is indeed the case.

\begin{remark}[Low-/high-pass controlled phase gates]
    The $C^\c_{\leq K} \Ph(t)$ gate applies a phase to the first $K$ states and consequently has no target qubit.
    An example is provided in Section~\ref{sec:oraclemaxkcut}.
\end{remark}

\section{Main theorem}
\label{sec:main}

We begin by introducing an important class of subspaces that will be central to the results that follow.
Let a group $G$ act on a set $X$, then the \textit{orbit} of an element $x \in X$ under this action is defined as
\begin{equation*}
G \cdot x = \{ g \cdot x \mid g \in G \}.
\end{equation*}
In the following we will consistently drop the $\cdot$ in the notation.
We define the following class of subspaces.
\begin{definition}[Pauli X-orbit subspaces]
Let ${ \lX_1, \ldots, \lX_k }$ with $\lX_i \in {I, X}^n$ be a minimal generating set of the group $G_k = \langle \lX_1, \ldots, \lX_k \rangle$.
Then the subspace generated by an element of the set of computational basis states, $\ket{z}$, is defined by the linear span of the orbit under $G_k$ as $\Span{G_k \ket{z}}$.  
\end{definition}
Note, that $|G_k \ket{z}|=2^k$ and that 
$G_k \ket{z}$ consists entirely of computational basis states. Moreover, the structure of these orbits satisfies the recursion
\begin{equation}
\label{eq:ggs_structure}
G_{k+1} \ket{z} = G_k \ket{z} \oplus G_k \lX_{k+1} \ket{z},
\end{equation}
where $\oplus$ denotes the direct sum of subspaces.
With this definition in place, we are now ready to state our main theorem.

\begin{theorem}
    \label{theorem:main}
    Let $P_B$ be the projector onto a subspace given by set of computational basis states $B$, and
    let $\sigma \in S_n$ be a Pauli string fulfilling $[\sigma, P_B]=0$.
    Then we can realize $e^{it H}$ for $H=\sigma P_B$ with a circuit consisting of
    \begin{itemize}
        \item at most two operators $M$ that permute states in $B$ to the first $K$ states together with $M^\dag$, and
        \item at most two low-pass controlled unitaries $C_{\leq K}^\c U_\t$ for $K\leq \nnze{}$ and $U_\t \in SU(2)$,
    \end{itemize}
    with the structure shown in Figure~\ref{fig:overview}.
    
    For a \textbf{general subspace} the permutation $M$ can be realized with $\Order{\nnze{}}$ transposition, so in total this means we need $\Order{n \nnze{}}$ CX gates and $\Order{n \nnze{} + \log(\nnze{}) \log (\frac{1}{\epsilon})}$ T gates, where $\epsilon$ is a given tolerance for approximating $\RZ$ gates.
   
    For a \textbf{Pauli X-orbit subspace} the permutation $M$ can be realized with $\Order{n \lognnze{}}$ CX gates with depth of $\Order{\log(n)\lognnze{}}$ and one $ C_{x}^{\c} R_{\sigma,\t}(\pm t)$, which can be realized with $\Order{n}$ CX gates and $\Order{n+\log(\frac{1}{\epsilon})}$ T gates.

    In both cases, the circuit depth scales as the number of T gates.
\end{theorem}

The idea of the construction/proof of Theorem~\ref{theorem:main} is to partition $B$ into sets $B_i$, where $\sigma$ acts as a unitary $U \in SU(2)$ on each set $B_s$.
The circuit to realize $e^{it H}$ consist of first applying a basis change $M_s$ for each of the sets $B_s$ that maps to the first $|B_s|$ computational basis states,
apply the unitary $U$, conditional on being in the first $|B_s|$ states, and apply the inverse basis change $M_s^\dag$,
as depicted in Figure~\ref{fig:overview}.
The proof for the case when $B$ is a Pauli X-orbit is presented in Section~\ref{sec:GroupGenerated}, and the general case in Section~\ref{sec:generalcase}.

\begin{remark}[Direct sum of Pauli X-generated subspaces]
    \label{remark:directsumgg}
    Even though $\Span{B}$ might not be a Pauli X-orbit subspace, it can be decomposed into a direct sum of Pauli X-orbit subspaces. As long as the number of those subspaces stays sufficiently small, this can still allow for an efficient construction.    
\end{remark}

\subsection{Proof for Pauli X-orbit generated subspaces}
\label{sec:GroupGenerated}

An important special case, where the subspace $B$ is given by a Pauli X-orbit of an $\ket{z} \in B$, we proof that there exists a permutation matrix $M$ consisting only of $X$ and $CX$ gates.

\begin{theorem}
\label{theorem:GroupGeneratedCase}
    Let $B=G_k\ket{z}$ be a Pauli X-orbit of $G_k = \langle \lX_1, \ldots, \lX_k \rangle$
    with $\Span{B} = 2^k$ and $P_B$ the projector onto $B$.
    \begin{enumerate}
        \item 
    Then there exists a permutation operator $M$ consisting of $\Order{n k}$ $X,CX$ gates such that
    \begin{equation}
        \label{eq:MPM_ggs}
        M P_B M^\dag = I_k \otimes \ket{0}\bra{0}_{n-k},
    \end{equation}
    where $I_k$ is the identity operator on $k$ qubits, and $\ket{0}\bra{0}_{n-k} = \bigotimes_{i=1}^{n-k} \ket{0}\bra{0}$.
    \item
    Let $\sigma$ be a Pauli string that commutes with $P_B$.
    Then $\exists \ \widehat \sigma \in S_k$ such that
    \begin{equation}
        \label{eq:SMPM_ggs}
        \sigma P_B  = \pm M  \widehat \sigma \otimes \ket{0}\bra{0}_{n-k} M^\dag,
    \end{equation}
    i.e. one can propagate $\sigma$ through $M$ consisting only of $X,CX$ gates such that it acts at most as a sign change on the last $n-k$ qubits.
    \end{enumerate}
\end{theorem}

The proof is constructive and leads directly to an efficient algorithm to construct the circuit for both~\eqref{eq:MPM_ggs} and~\eqref{eq:SMPM_ggs}.
\begin{proof}

We will start by proving the first assertion by induction.

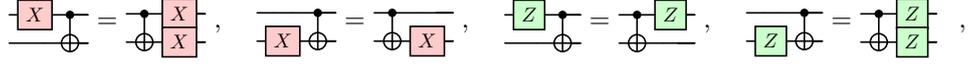
\begin{figure}
    \centering
    \begin{tikzpicture}[remember picture, baseline={0cm}]
    \node[scale=.75] {
    \begin{quantikz}[row sep={0.5cm,between origins},column sep=1ex]
         & \gate[style={fill=red!20}]{X} & \ctrl{1} & \qw \\
         & \phantomgate{X} & \targ{} & \qw
    \end{quantikz}
    };
\end{tikzpicture}
    \kern-0.5em  = \kern-0.5em
\begin{tikzpicture}[remember picture, baseline={0cm}]
    \node[scale=.75] {
    \begin{quantikz}[row sep={0.5cm,between origins},column sep=1ex]
        & \ctrl{1} & \gate[style={fill=red!20}]{X} & \qw \\
        & \targ{} & \gate[style={fill=red!20}]{X} & \qw
    \end{quantikz}
    };
\end{tikzpicture}
    \kern-0.5em , \kern0.5em
\begin{tikzpicture}[remember picture, baseline={0cm}]
    \node[scale=.75] {
    \begin{quantikz}[row sep={0.5cm,between origins},column sep=1ex]
         & \phantomgate{X} & \ctrl{1} & \qw \\
         & \gate[style={fill=red!20}]{X} & \targ{} & \qw
    \end{quantikz}
    };
\end{tikzpicture}
    \kern-0.5em  = \kern-0.5em
\begin{tikzpicture}[remember picture, baseline={0cm}]
    \node[scale=.75] {
    \begin{quantikz}[row sep={0.5cm,between origins},column sep=1ex]
        & \ctrl{1} & \phantomgate{X} & \qw \\
        & \targ{} & \gate[style={fill=red!20}]{X} & \qw
    \end{quantikz}
    };
\end{tikzpicture}
    \kern-0.5em , \kern0.5em
\begin{tikzpicture}[remember picture, baseline={0cm}]
    \node[scale=.75] {
    \begin{quantikz}[row sep={0.5cm,between origins},column sep=1ex]
         & \gate[style={fill=green!20}]{Z} & \ctrl{1} & \qw \\
         & \phantomgate{Z} & \targ{} & \qw
    \end{quantikz}
    };
\end{tikzpicture}
    \kern-0.5em  = \kern-0.5em
\begin{tikzpicture}[remember picture, baseline={0cm}]
    \node[scale=.75] {
    \begin{quantikz}[row sep={0.5cm,between origins},column sep=1ex]
        & \ctrl{1} & \gate[style={fill=green!20}]{Z} & \qw \\
        & \targ{} & \phantomgate{Z} & \qw
    \end{quantikz}
    };
\end{tikzpicture}
    \kern-0.5em , \kern0.5em
\begin{tikzpicture}[remember picture, baseline={0cm}]
    \node[scale=.75] {
    \begin{quantikz}[row sep={0.5cm,between origins},column sep=1ex]
         & \phantomgate{Z} & \ctrl{1} & \qw \\
         & \gate[style={fill=green!20}]{Z} & \targ{} & \qw
    \end{quantikz}
    };
\end{tikzpicture}
    \kern-0.5em  = \kern-0.5em
\begin{tikzpicture}[remember picture, baseline={0cm}]
    \node[scale=.75] {
    \begin{quantikz}[row sep={0.5cm,between origins},column sep=1ex]
        & \ctrl{1} & \gate[style={fill=green!20}]{Z} & \qw \\
        & \targ{} & \gate[style={fill=green!20}]{Z} & \qw
    \end{quantikz}
    };
\end{tikzpicture}
,
    \caption{Permutation rules for the CX gate and Pauli operators, often presented in the context of error propagation.
    }
    \label{fig:commrulesCX}
\end{figure}

\textbf{Base case $k=1$.}
The subspace is given by $G_1\ket{z} = \{\ket{z}, \ket{w} = \lX_1 \ket{z} \}$.
Applying a Pauli $X$ gate on the indices where $z$ is $1$ we are left with the states $\ket{0 \cdots 0}$ and $\ket{z\oplus w}$.
Looking at the indices where $z\oplus w$ is $1$, we see that we can use the ``$CX$-stairs'' from the GHZ state preparation circuit to map $\ket{0 \cdots 0}$ to itself and $\ket{z\oplus w}$ to $\ket{1 0 \cdots 0}$, where we relabel the indices in case the $1$ in $z \oplus w$ is not in the first index.
In total, we get the form shown in Equation~\eqref{eq:MPM_ggs} with the identity operator on the first qubit. We need at most $\Order{n}$ X, CX gates to realize $M$. 

\textbf{Induction step $k \rightarrow k+1$.}
Let the assumption hold for $k$, i.e. $\exists M_k$ constructed by $X, CX$ gates, such that for $G_k = \langle \lX_1, \cdots, \lX_k\rangle$ we have
\begin{equation}
    \label{eq:G_k}
    M_k \sum_{\ket{w} \in G_k\ket{z}} \ket{w} \bra{w} M_k^\dag = I_k \otimes \ket{x}\bra{x}.
\end{equation}
Since $M_k$ can be realized with only $X, CX$ gates, we know that we can propagate Pauli-X gates through $M_k$.
From this it follows trivially that
\begin{equation}
    \label{eq:XG_k}
    M_k \sum_{\ket{w} \in G_k\ket{z}} (\lX_{k_1}\ket{w} \bra{w}\lX_{k_1}) M_k^\dag = \hat{\lX}_{k_1} I_k \otimes \ket{x}\bra{x}\hat{\lX}_{k_1} = I_k \otimes \ket{\hat{x}}\bra{\hat{x}},
\end{equation}
where $\hat{x}$ is a bit-string.
We can add Equations~\eqref{eq:G_k} and ~\eqref{eq:XG_k}, using Equation~\eqref{eq:ggs_structure} to see that
\begin{equation*}
    M_k \sum_{\ket{w} \in G_{k+1}\ket{z}} \ket{w} \bra{w} M_k^\dag = I_k \otimes (\ket{x}\bra{x} + \ket{\hat{x}}\bra{\hat{x}}).
\end{equation*}
Using the base case, we know there exists an $\widetilde M$ which can be realized with only X and CX gates, such that 
$\widetilde M I_k \otimes \left(\ket{x}\bra{x} + \ket{\hat{x}}\bra{\hat{x}}\right)\widetilde M^\dag = I_{k+1} \otimes \ket{y}\bra{y}$.
Setting $M_{k+1} = M_k \widetilde M$ shows the existence of the permutation operator.
Overall, the circuit for $M$ can be realized with
\begin{equation}
\label{eq:Mghz}
    M = \begin{quantikz}[row sep={0.5cm,between origins},column sep=1ex, wire types={n,n,n,n,n,n}]
         \arrow[arrows]{rrrrr}
               & \gate[wires=6,style={text height=3cm, text width=.65cm,text depth=0}]{X^{\z_1}}
                    & \gate[wires=6,style={text height=3cm, text width=1.3cm,text depth=0}]{M^{\z_1\oplus \w_1}_{GHZ}}
                        & \gate[wires=5,style={text height=2.5cm, text width=.65cm,text depth=0}]{X^{\z_2}}
                            & \gate[wires=5,style={text height=2.5cm, text width=1.3cm,text depth=0}]{M^{\z_2\oplus \w_2}_{GHZ}}
                                &\ldots
                                        & \gate[wires=3,style={text height=1.5cm, text width=.65cm,text depth=0}]{X^{\z_k}}
                                            & \gate[wires=3,style={text height=1.5cm, text width=1.3cm,text depth=0}]{M^{\z_k\oplus \w_k}_{GHZ}}
                                                & \arrow[arrows]{lll}\\
         \vdots&    &   &   &   &       &   &   & \\
         \arrow[arrows]{rrrrr}
               &    &   &   &   &\ldots &   &   & \arrow[arrows]{lll} \\
         \vdots&    &   &   &   &\raisebox{2mm}{$\adots$}
                                        &   &   & \\
         \arrow[arrows]{rrrrrrrr}
               &    &   &   &   &       &   &   & \\
         \arrow[arrows]{rrrrrrrr}
               &    &   &   &   &       &   &   &
    \end{quantikz},
\end{equation}
where $M^\z_{GHZ}$ is the circuit realizing the GHZ state for the states where $\z$ is one, \textit{without} the Hadamard gate.
Note, that each individual circuit for $M^\z_{GHZ}$ can be realized with logarithmic depth in the number of qubits~\cite{cruz2019efficient}.
Of course, all $X$ gates can be permuted to the front.
In total, 
the number of $X, CX$ gates scales as
$
\sum_{j=0}^{k-1} (n - j) = k \cdot n - \frac{k(k - 1)}{2}
$,
and the depth scales as 
$
\sum_{j=0}^{k-1} \log(n - j) = \log \left( \frac{n!}{(n - k)!} \right)
$.

We will now proof the second assertion. 
Since $M$ is realized with $X,CX$ gates, there exists a Pauli string $\widetilde\sigma$ such that
$\sigma M = M \widetilde\sigma$. See Figure~\ref{fig:commrulesCX} showing the rules that can be applied for the CX-gate.
Using the commutation relation, we have that
\begin{equation*}
    \left( I_k \otimes \ket{x} \bra{x}_B\right) \widetilde \sigma = M^\dag P \sigma M =
    M^\dag \sigma P M = \widetilde \sigma \left(I_k \otimes \ket{x} \bra{x}_B\right) .
\end{equation*}
Since $\ket{x}$ is not the zero vector, it follows directly that $\restr{\widetilde \sigma}{B} \ket{x}\bra{x} =\pm \ket{x}\bra{x}$.
\end{proof}

\begin{corollary}
    Let $B=G_k\ket{z}$ so that $\Span{B}$ is a Pauli X-orbit generated subspace of dimension $2^k$ and $P_B$ its projector.
    Let $\sigma$ be a Pauli string that commutes with $P_B$.
    Then there exists a permutation operator $M$ and a Pauli string $\sigma\in S_k$ such that
    \begin{equation*}
        e^{it\sigma P_B} =
        M C_{x}^{\c} R_{\sigma,\t}(\pm t) M^\dag,
    \end{equation*}
    with control $\c=(k+1,\ldots,n)$ and target $\t= (1,\ldots, k)$. 
    The permutation $M$ can be realized with $\Order{n \lognnze}$ $X,CX$ gates with circuit depth $\Order{\log(n)\lognnze}$.
    The multi-controlled $\sigma$ gate can be realized with $\Order{n}$ $X,CX$ gates with circuit depth $\Order{n}$ and $\Order{\log(1/\epsilon)}$ T gates.
\end{corollary}
\begin{proof}
The assertion follows directly from Theorem~\ref{theorem:GroupGeneratedCase}, which gives us the existence of $M$ consisting of $\Order{n \lognnze}$ $X,CX$ gates with depth $\Order{\log(n) \lognnze}$ depth such that
\begin{equation*}
        e^{it \sigma P_B}  = 
        M  e^{i t (\pm \widehat \sigma) \otimes \ket{\x}\bra{\x}} M^\dag,
\end{equation*}
where $\sigma$ is a Pauli string acting on $k$ qubits.
The exponential of $(\pm \widehat \sigma) \otimes \ket{\x}\bra{\x}$ can be realized with a $\sigma$ rotation controlled by the state $\ket{\x}$.
\end{proof}

\subsection{Proof for general subspaces}
\label{sec:generalcase}

We divide the proof of Theorem~\ref{theorem:main} for a general subspace into four cases as follows.

    \subsubsection{Case \texorpdfstring{$\sigma=I$}{I}}
    We start by mapping the (indexed) states of $B$ to the first $\nnze{}$ computational basis states
    which can be achieved by
    \begin{equation*}
        M=\prod_{j=0}^{\nnze-1} T_{j,z_j}.
    \end{equation*}
    It follows directly that
    \begin{equation*}
        M P_B M^\dagger \ket{j} = M P_B \ket{z_j} = 
        \begin{cases}
            M \ket{z_j} = \ket{j},& \text{ for } j< \nnze{},\\
            \mathbf{0},& \text{ otherwise.}
        \end{cases}
    \end{equation*}
    Therefore, we have that
    $
        MP_{B}M^{\dag} 
        = \sum_{j< \nnze{}}\ket{j}\bra{j},
    $
    from which it follows that the exponential has the form
    \begin{equation*}
        e^{itP_B} =
        M^\dag e^{it\sum_{j< \nnze{}}\ket{j}\bra{j}} M =
        M^\dag C_{\leq \nnze{}}^\c\Ph(t) M
    \end{equation*}
    where $\c=(1,\ldots, n)$.
    
    \subsubsection{Case \texorpdfstring{$\sigma=Z^\b$}{Z}}
    For $s\in\{+,-\}$ we define
    \begin{equation*}
        B_s \coloneqq \{ \ket{\z} \in B \ | \ Z^\b\ket{\z} = s\ket{\z} \}.
    \end{equation*}
    Furthermore, after indexing the states in $B_s$ we define the permutation matrices
    \begin{equation*}
        M_s = \prod_{\substack{j =0,\\ \ket{\z_j}\in B_s}}^{|B_s|-1} T_{j,z_j}.
    \end{equation*}
    It follows that
    \begin{equation*}
        M_s Z^\b P_{B_s} M^\dag_s \ket{j} = 
        M_s Z^\b P_{B_s} \ket{z_j} = 
        \begin{cases}
            M_s Z^\b \ket{z_j} = s \ket{j},& \text{ for } j< |B_s|,\\
            \mathbf{0},& \text{ otherwise.}
        \end{cases}
    \end{equation*}
    Write $P_B=P_{B_+}+P_{B_-}$ and using that $[Z^\b P_{B_-},Z^\b P_{B_+}]=0$ since it involves diagonal matrices,
    we have
    \begin{equation*}
        e^{it Z^\b P_B} =
        \prod_{s\in\{+,-\} } e^{it Z^\b P_{B_s}} =
        \prod_{s\in\{+,-\} } M_s^\dag e^{it P_{B_s}}M_s =
        \prod_{s\in\{+,-\} } M_s^\dag C_{\leq |B_s|}^\c\Ph(s t) M_s,
    \end{equation*}
    which proofs the assertion.
    
    \subsubsection{Case \texorpdfstring{$X^\a \neq I,\left [ X_\sigma^{\alpha},Z^{\beta}_\sigma\right]\neq 0$}{X not commuting}}
    We define the set of ordered pairs
    \begin{equation*}
        E \coloneqq \{ (\ket{x},\ket{y}) \in B\times B \ | \ X^\a\ket{x} = \ket{y}, Z^\b\ket{x} = +\ket{x} \}.
    \end{equation*}
    For $(\ket{x},\ket{y})\in E$ we have that
    $Z^\b \ket{y} = Z^\b X^\a \ket{x} = - X^\a Z^\b \ket{x} = X^\a \ket{x}=-\ket{y}$ which is the defining property for the ordering of the pairs in $E$.
    The matrix $\sigma$ acts on $(\ket{\x},\ket{\y})\in E$ in the following way
    \begin{equation*}
        \begin{split}
            \sigma \ket{x} &= i^{\a\cdot \b} X^\a Z^\b \ket{x} = 
            i^{\a\cdot \b} X^\a \ket{x} = 
            i^{\a\cdot \b} \ket{y},\\
            \sigma \ket{y} &= i^{\a\cdot \b} X^\a Z^\b \ket{y} = 
            -i^{\a\cdot \b} X^\a \ket{x} = 
            -i^{\a\cdot \b} \ket{x}.
        \end{split}
    \end{equation*}
    Again for $(\ket{\x},\ket{\y})\in E$ we have $\ket{x} = \sigma^2 \ket{x} = \sigma i^{\a\cdot \b} \ket{y} =  -  i^{2 (\a\cdot \b)} \ket{x}$, from which it follows that $\a \cdot \b = 1 \!\! \pmod{4}$ or $\a \cdot \b = 3 \!\! \pmod{4}$,
    which means that $i^{\a\cdot \b} = \pm i$.
    After indexing the pairs in $E$
    we define the permutation matrix
    \begin{equation*}
        M = \prod_{\substack{j=0,\\ (\ket{\x_j},\ket{\y_j})\in E}}^{|E|-1} T_{2j,\x_j} T_{2j+1, \y_j}.
    \end{equation*}
    From this it follows that
    \begin{equation*}
        M \sigma P_B M^\dag \ket{j}
        = 
        \begin{cases}
            M \sigma \ket{x_j} = \pm i M \ket{y_j} = \pm i \ket{j+1},& \text{ for } j< |E|, \ j \! \! \! \! \pmod{2} = 0,\\ 
            M \sigma \ket{y_j} = \mp i M \ket{x_j} = \mp i \ket{j-1},& \text{ for } j< |E|, \ j \! \! \! \! \pmod{2} = 1,\\ 
            \mathbf{0},& \text{ for } j\geq|E|,
        \end{cases}
    \end{equation*}
    i.e. $
        M \sigma P_B M^\dag = \pm \sum_{j=0}^{|E|-1} i \left( \ket{2j}\bra{2j+1} -  \ket{2j+1}\bra{2j}\right)
        = \pm \sum_{j<\nnze{}/2}\ket{j}\bra{j}_{n-1}\otimes Y 
    $.
    Hence, the exponential has the form
    \begin{equation*}
        e^{it\sigma P_B} =
        M^\dag e^{it\sum_{j<\nnze{}/2} \ket{j}\bra{j}_{n-1}\otimes (\pm Y) )} M =
        M^\dag C_{\leq \nnze{}/2}^\c\RY[\t](\pm t) M
    \end{equation*}
    where $\c=(1,\ldots, n-1)$, and $\t=(n)$.

    \subsubsection{Case \texorpdfstring{$X^\a \neq I,\left [ X_\sigma^{\alpha},Z^{\beta}_\sigma\right]=0$}{X commuting}}
    We define the set of unordered pairs
    \begin{equation*}
        E_s \coloneqq \{ \{\ket{x},\ket{y}\} \in B\times B \ | \ X^\a\ket{x} = \ket{y}, Z^\b\ket{x} = s\ket{x} \},
    \end{equation*}
    where $s\in\{+,-\}$.
    The matrix $\sigma$ acts on $\{\ket{\x},\ket{\y}\}\in E_s$ in the following way
    \begin{equation*}
            \sigma \ket{x} = i^{\a\cdot \b} X^\a Z^\b \ket{x} = 
            s i^{\a\cdot \b} \ket{y},
    \end{equation*}
    and likewise for
    $\sigma \ket{y} = s i^{\a\cdot \b} \ket{x}$.
    Furthermore, for $(\ket{\x},\ket{\y})\in E_s$ we have $\ket{x} = \sigma^2 \ket{x} = s \sigma i^{\a\cdot \b} \ket{y} =  i^{2 (\a\cdot \b)} \ket{x}$, from which it follows that $\a \cdot \b = 0 \!\! \pmod{4}$ or $\a \cdot \b = 2 \!\! \pmod{4}$,
    which means that $i^{\a\cdot \b} = \pm 1$.
    After indexing the pairs in $E$
    we define the permutation matrices
    \begin{equation*}
        M_s = \prod_{\substack{j=0,\\ \{\ket{\x_j},\ket{\y_j}\}\in E_s}}^{|E_s|-1} T_{2j,\x_j} T_{2j+1, \y_j}.
    \end{equation*}
    From this it follows that
    \begin{equation*}
        M_s \sigma P_{B_s} M_s^\dag \ket{j}
        = 
        \begin{cases}
            M \sigma \ket{x_j} = \pm s M \ket{y_j} = \pm s \ket{j+1},& \text{ for } j< |E|, \ j \! \! \! \! \pmod{2} = 0,\\ 
            M \sigma \ket{y_j} = \pm s M \ket{x_j} = \pm s \ket{j-1},& \text{ for } j< |E|, \ j \! \! \! \! \pmod{2} = 1,\\ 
            \mathbf{0},& \text{ for } j\geq|E|,
        \end{cases}
    \end{equation*}
    i.e. $
        M_s \sigma P_{B_s} M_s^\dag = \pm s \sum_{j=0}^{|E_s|-1} \left( \ket{2j}\bra{2j+1} +  \ket{2j+1}\bra{2j}\right)
        = \pm s \sum_{j<|B_s|/2}\ket{j}\bra{j}_{n-1}\otimes X 
    $.
    Define
    $
        B_\pm = \{\ket{x} | \{\ket{x},\ket{y}\} \in E_\pm \}\}.
        $
    and write $P_B=P_{B_+}+P_{B_-}$. Note, that $B_+\cap B_- = \{\}$.
    Observe that
    \begin{equation*}
        \sigma P_{B_s}
        = 
        \sigma\sum_{\{\ket{x},\ket{y}\} \in E_s} \left(\ket{x}\bra{x} + \ket{y}\bra{y}\right)
        = 
        \pm s \sum_{\{\ket{x},\ket{y}\} \in E_s} \left(\ket{y}\bra{x} + \ket{x}\bra{y}\right)
        = 
        P_{B_s} \sigma,
    \end{equation*}
    which leads to $ 
    [\sigma P_{B_-}, \sigma P_{B_+}] = \sigma[P_{B_-}, \sigma] P_{B_+} + \sigma\sigma[P_{B_-}, P_{B_+}] + [\sigma,\sigma] P_{B_+}P_{B_-} + \sigma[\sigma,  P_{B_+}]P_{B_-}
    =0$.
    Therefore, we have that
    \begin{equation*}
        e^{it \sigma P_B} =
        \prod_{s\in\{+,-\} } e^{it \sigma P_{B_s}} =
        \prod_{s\in\{+,-\} } M_s^\dag C_{\leq |B_s|/2}^\c\RX[\t](\pm s t) M_s,
    \end{equation*}
    where $\c=(1,\ldots, n-1)$, and $\t=(n)$.

This concludes the proof and we continue with example applications.

\section{Examples}
\label{sec:examples}

\subsection{Transposition gates}

As an example we can apply Theorem~\ref{theorem:GroupGeneratedCase} to the transposition gate which can be written as,
\begin{equation*}
     T_{\x,\y} = X^{\x\oplus \y} P_{\{\x,\y\}},
\end{equation*}
i.e. we can apply Theorem~\ref{theorem:GroupGeneratedCase} for $\sigma = X^{\x\oplus \y}$ and $B=\langle X^{\x\oplus \y} \rangle \ket{\x}$.
Since the subspace is generated by a Pauli X-orbit, according to Equation~\eqref{eq:Mghz} $M = X^{\x} M^{\x\oplus \y}_{GHZ}$ consists of at most $\Order{n}$ X and CX gates and has depth $\Order{\log(n)}$.
Compared to the method proposed in \cite{herbert2024almost}
this approach reduces the resources to realize the transposition gate, by not requiring an ancilla qubit and using only one multi-controlled Toffoli gate (with one less control).
We present two typical examples for $4$ and $3$ qubits realizing a transposition gate between two computational basis states through
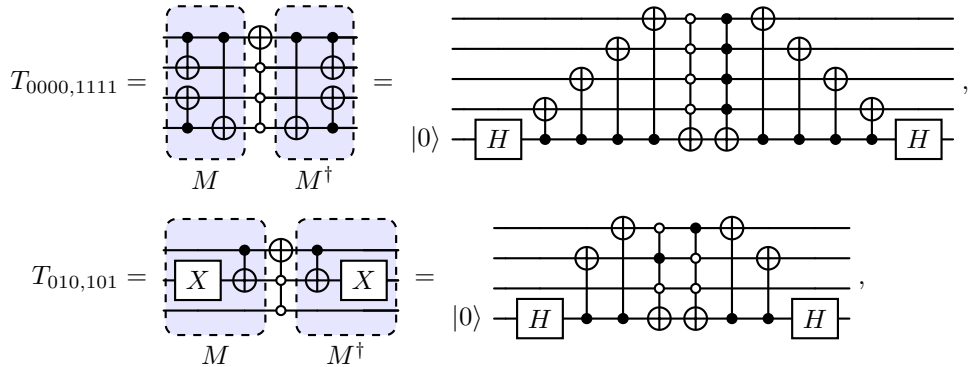
\begin{equation*}
\begin{split}
    T_{0000,1111}&=
    \begin{quantikz}[row sep={0.4cm,between origins},column sep=1ex]
    \qw & \ctrl{1}
    \gategroup[4,steps=2,style={dashed,rounded corners,fill=blue!10, inner xsep=0pt},background, label style={label position=below, yshift=-0.5cm}]{
    $M$
    }
                    & \ctrl{3} & \targ{}    & \ctrl{3}
    \gategroup[4,steps=2,style={dashed,rounded corners,fill=blue!10, inner xsep=0pt},background, label style={label position=below, yshift=-0.5cm}]{
    $M^\dag$
    }
                                                       & \ctrl{1}  & \qw \\
    \qw & \targ{}   & \qw      & \octrl{-1} & \qw      & \targ{}   & \qw \\
    \qw & \targ{}   & \qw      & \octrl{-1} & \qw      & \targ{}   & \qw \\
    \qw & \ctrl{-1} & \targ{}  & \octrl{-1} & \targ{}  & \ctrl{-1} & \qw 
    \end{quantikz}
    =
    \begin{quantikz}[row sep={0.4cm,between origins},column sep=1ex]
     & &          &           &           &           & \targ{}   & \octrl{1} & \ctrl{1} & \targ{}   &           &           &           &          &\\
     &  &          &           &           & \targ{}   &           & \octrl{1} & \ctrl{1} &           & \targ{}   &           &           &          &\\
     & &          &           & \targ{}   &           &           & \octrl{1} & \ctrl{1} &           &           & \targ{}   &           &          &\\
     & &          & \targ{}   &           &           &           & \octrl{1} & \ctrl{1} &           &           &           & \targ{}   &          &\\
     \lstick{$\ket{0}$}
     & & \gate{H} & \ctrl{-1} & \ctrl{-2} & \ctrl{-3} & \ctrl{-4} & \targ{}   & \targ{}  & \ctrl{-4} & \ctrl{-3} & \ctrl{-2} & \ctrl{-1} & \gate{H}& 
    \end{quantikz},
    \\
    T_{010,101}&=
    \begin{quantikz}[row sep={0.4cm,between origins},column sep=1ex]
    \qw &  
    \gategroup[3,steps=2,style={dashed,rounded corners,fill=blue!10, inner xsep=0pt},background, label style={label position=below, yshift=-0.5cm}]{
    $M$
    }
                    & \ctrl{1} & \targ{}    & \ctrl{1}
    \gategroup[3,steps=2,style={dashed,rounded corners,fill=blue!10, inner xsep=0pt},background, label style={label position=below, yshift=-0.5cm}]{
    $M^\dag$
    }
                                                       &           & \qw \\
    \qw & \gate{X}  & \targ{}  & \octrl{-1} & \targ{}  & \gate{X}  & \qw \\
    \qw &           &          & \octrl{-1} &          &           & \qw 
    \end{quantikz}
    =
    \begin{quantikz}[row sep={0.4cm,between origins},column sep=1ex]
     & &         &           & \targ{}   & \octrl{1} & \ctrl{1}  & \targ{}   &           &           &\\
     & &         & \targ{}   &           & \ctrl{1}  & \octrl{1} &           & \targ{}   &           &\\
     & &         &           &           & \octrl{1} & \octrl{1} &           &           &           &\\
     \lstick{$\ket{0}$}
     & & \gate{H} & \ctrl{-2}& \ctrl{-3} & \targ{}   & \targ{}   & \ctrl{-3} & \ctrl{-2} & \gate{H} & 
    \end{quantikz},
\end{split}
\end{equation*}
where the left is the proposed method, and the right is the one from~\cite{herbert2024almost}.

\subsection{Fermionic excitations}
The unitary coupled cluster ansatz~\cite{Hoffmann1988} for simulating molecules
requires to realize fermionic excitation operators defined by the exponential of the skew-Hermitian operators
\begin{equation*}
    T_i^k = a^\dag_k a_i - a^\dag_i a_k, \quad 
    T_{i,j}^{k,l} = a^\dag_k a^\dag_l a_i a_j - a^\dag_i a^\dag_j a_k a_l, \quad 
    \cdots, \quad
    T_{i_1, \cdots, i_n}^{k_1,\cdots,k_n} = 
    \prod_{j=1}^n a^\dag_{k_j} a_{i_j}
    -
    \prod_{j=1}^n a^\dag_{i_j} a_{k_j}
\end{equation*}
related to single, double, and higher excitation operators, respectively.
Here, $a_i^{\dag}$ and $a_i$ refer to the fermionic creation and annihilation operators.
The Jordan-Wigner mapping is given by
$a_i=Q_i\prod_{i=0}^{r-1}Z_r, \quad a_i^\dag=Q_i^\dag\prod_{i=0}^{r-1}Z_r$
with the qubit creation and annihilation operators defined as
$Q_i^{\dag}=\frac{1}{2}(X_i-iY_i)$ and $Q_i=\frac{1}{2}(X_i+iY_i)$.
In this case,
the $n$ qubit excitation operators can be expressed as 
\begin{equation}
\label{eq:nqubitexcitation}
    \widehat T_{i_1, \cdots, i_n}^{k_1,\cdots,k_n} =
    \prod_{j=1}^n Q^\dag_{k_j} Q_{i_j}
    -
    \prod_{j=1}^n Q^\dag_{i_j} Q_{k_j}
    =
    -i Z_{i_1} \prod_{j=1}^{n} X_{i_j} X_{k_j} G,
\end{equation}
where $G$ is the group generated by the minimal set of Pauli operators
$
    \left\{ Z_{i_j} Z_{i_{j+1}},\ Z_{k_j} Z_{k_{j+1}} \mid 1 \leq j \leq n - 1 \right\}
    \cup 
    \left\{ -Z_{i_n} Z_{k_1} \right\}
$,
and we define $\sigma H \coloneqq \allowbreak \frac{1}{|H|} \sum_{h\in H} \sigma h$ for a group $H = \langle \sigma_1, \cdots, \sigma_k\rangle$. 
Interpreting this through the lens of stabilizer codes used in quantum error correction, it is easy to see that the subspace of the projector $P_B$ is given by the
stabilizer subspace $B=\{ \ket{0\dots01\dots 1}, \allowbreak \ket{1\dots10\dots 0}\}$, and $\sigma$ is any logical $Y$-operator times $i$, e.g. 
for $X^{\alpha}_\sigma=X\cdots X$ and $Z^{\beta}_{\sigma}=ZI\cdots I$ in symplectic from.

The \textbf{n-qubit excitation operator} expressed in the Pauli basis as in Equation~\eqref{eq:nqubitexcitation} has $2^{2n-1}$ non-zero Pauli strings, meaning one needs an exponential number of $\RZ$ and therefore also T gates when the circuit is realized through Pauli evolution.
On the other hand, since the subspace is generated by a Pauli X-orbit, applying the proposed method yields a circuit composed of $M = X^{\x} M^{\mathbf{1}}_{GHZ}$ for $\x = 0\dots01\dots 1$ according to Equation~\eqref{eq:Mghz} and one
$n-1$ controlled $\RY$ gate. This leads to a circuit with $\Order{n}$ X, CX gates and $\Order{n+\log(\frac{1}{\epsilon})}$ T-gates for realizing the n-qubit excitation operator.

Note, that the fermionic unitaries can be efficiently obtain from the unitary evolution of qubit excitation operators~\cite{yordanov2020efficient}, using the circuit realization depicted in Figure~\ref{fig:Pauliexp}, so that it is sufficient to be able to realize qubit excitation operators efficiently.
 
As an example, one can realize the parametrized first qubit excitation operator through the circuit
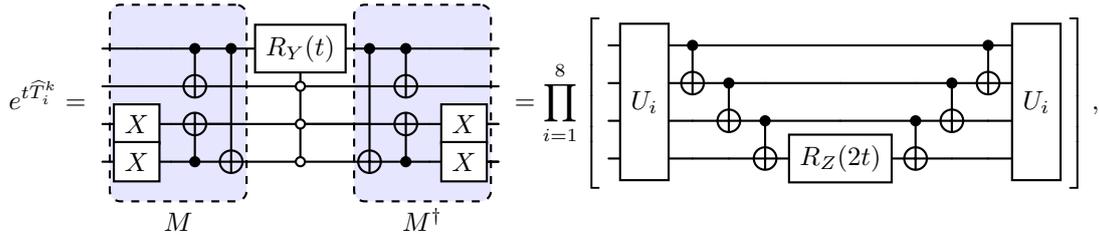
\begin{align*}
    e^{t \widehat T_i^k} &
    =
    \begin{quantikz}[row sep={0.5cm,between origins},column sep=1ex]
        \qw & 
    \gategroup[4,steps=4,style={dashed,rounded corners,fill=blue!10, inner xsep=-2pt},background, label style={label position=below, yshift=-0.5cm}]{
    $M$
    }
                   &    & \ctrl{1}  & \ctrl{3} & \gate{\RY(t)}    & \ctrl{3}
    \gategroup[4,steps=4,style={dashed,rounded corners,fill=blue!10, inner xsep=-2pt},background, label style={label position=below, yshift=-0.5cm}]{
    $M^\dag$
    }
                                                                        & \ctrl{1}  & & &\qw \\
        \qw &        & & \targ{}   & \qw      & \octrl{-1} & \qw      & \targ{}   & & &\qw \\
        \qw &\gate{X} & & \targ{}   & \qw      & \octrl{-1} & \qw      & \targ{}   & &\gate{X}&\qw \\
        \qw &\gate{X} & & \ctrl{-1} & \targ{}  & \octrl[    ]{-1} & \targ{}  & \ctrl{-1} & &\gate{X} &\qw 
    \end{quantikz}
    = \prod_{i=1}^8
    \left [
    \begin{quantikz}[row sep={0.5cm,between origins},column sep=1ex]
        & \gate[4]{U_i} & \ctrl{1} &          &          &                    &          &          & \ctrl{1} & \gate[4]{U_i}  \\
        &             & \targ{}  & \ctrl{1} &          &                    &          & \ctrl{1} & \targ{}  &              \\
        &             &          & \targ{}  & \ctrl{1} &                    & \ctrl{1} & \targ{}  &          &              \\
        &             &          &          & \targ{}  & \gate{\RZ( 2t)} & \targ{}  &          &          &             
    \end{quantikz}
    \right ]
    ,
    \\
\end{align*}
and second qubit excitation operator through
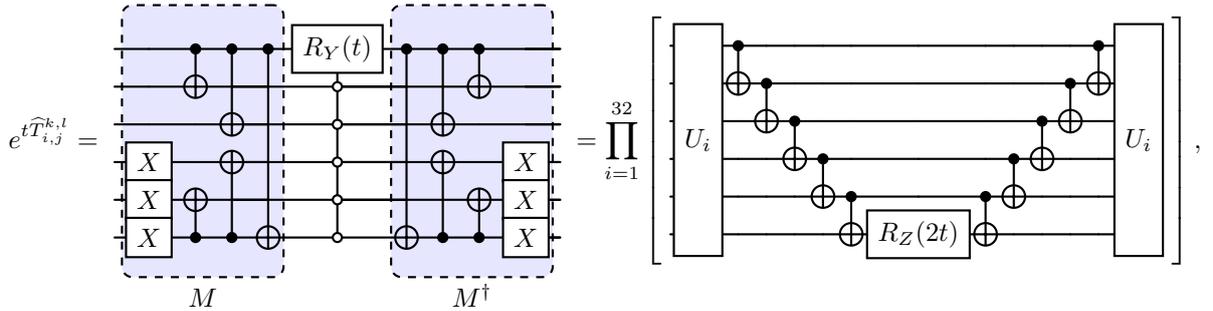
\begin{align*}
        e^{t \widehat T_{i,j}^{k,l}} &
    =
    \begin{quantikz}[row sep={0.5cm,between origins},column sep=1ex]
        \qw& 
    \gategroup[6,steps=4,style={dashed,rounded corners,fill=blue!10, inner xsep=-2pt},background, label style={label position=below, yshift=-0.5cm}]{
    $M$
    }
                  & \ctrl{1}  & \ctrl{2} &\ctrl{5} & \gate{\RY(t)}    & \ctrl{5}
    \gategroup[6,steps=4,style={dashed,rounded corners,fill=blue!10, inner xsep=-2pt},background, label style={label position=below, yshift=-0.5cm}]{
    $M^\dag$
    }
                                   &\ctrl{2} & \ctrl{1} &         & \qw \\
        \qw&         & \targ{}   &          &\qw      & \octrl{-1}       & \qw      &         & \targ{}  &         & \qw \\
        \qw&         &           & \targ{}  &\qw      & \octrl{-1}       & \qw      &\targ{}  &          &         & \qw \\
        \qw&\gate{X} &           & \targ{}  &\qw      & \octrl{-1}       & \qw      &\targ{}  &          &\gate{X} & \qw \\
        \qw&\gate{X} & \targ{}   &          &\qw      & \octrl{-1}       & \qw      &         & \targ{}  &\gate{X} & \qw \\
        \qw&\gate{X} & \ctrl{-1} & \ctrl{-2}&\targ{}  &\octrl[    ]{-1} & \targ{}   &\ctrl{-2}& \ctrl{-1}&\gate{X} & \qw 
    \end{quantikz}
        = \prod_{i=1}^{32}
    \left [
    \begin{quantikz}[row sep={0.5cm,between origins},column sep=.3ex]
        & \gate[6]{U_i} & \ctrl{1} &          &          &          &          &                    &          &          &          &          & \ctrl{1} & \gate[6]{U_i}  \\
        &               & \targ{}  & \ctrl{1} &          &          &          &                    &          &          &          & \ctrl{1} & \targ{}  &                \\
        &               &          & \targ{}  & \ctrl{1} &          &          &                    &          &          & \ctrl{1} & \targ{}  &          &                \\
        &               &          &          & \targ{}  & \ctrl{1} &          &                    &          & \ctrl{1} & \targ{}  &          &          &                \\
        &               &          &          &          & \targ{}  & \ctrl{1} &                    & \ctrl{1} & \targ{}  &          &          &          &                \\
        &               &          &          &          &          & \targ{}  & \gate{\RZ(2 t)} & \targ{}  &          &          &          &          &               
    \end{quantikz}
    \right ]
    ,
\end{align*}
where the left is the proposed method and the right is through Pauli evolution of all non-zero Pauli strings.

\subsection{Trace gate for lattice gauge theory }

Here we want to construct quantum circuits realizing the trace gate proposed in \cite{Alam_2022} for the Dihedral group $D_{N}$ when $N=2^n$. The trace gate operator is a diagonal operator of the following form :
\begin{equation*}
    H_{\Tr}^{2^n}=\ket{0}\bra{0}\otimes\sum_{k=0}^{2^n-1}\cos(\frac{2\pi k}{2^n})\ket{k}\bra{k}.
\end{equation*}
Observe, that for $k>0$ the two states
$\{k,2^{n-1}+k\}$ as well as the two states $\{2^{n}-k,2^{n-1}-k\}$ are connected with a flip of the first bit.
Hence, the subspace can be generated by
\begin{equation*}
    P_k = \langle X_1, X^{\operatorname{bin}(k)\oplus\operatorname{bin}(2^{n}-k)}\rangle \ket{k}.
\end{equation*}
Looking at the relative moduli of the four connected states, this allows us to rewrite
\begin{equation*}
    H_{\Tr}^{2^n}=
    \ket{0}\bra{0}\otimes Z \otimes \ket{0}\bra{0}_{n-1}
    +
    \ket{0}\bra{0}\otimes\sum_{k=1}^{2^{n-2}-1}\cos(\frac{2\pi k}{2^n})\sigma P_k
\end{equation*}
with $\sigma = Z\otimes Z\otimes I_{n-2}$ and $P_k = I_2\otimes \ket{k}\bra{k}_{n-2}$.
It is important to notice that this implementation scales as $2^{n-2}$ but implements the exact operator reducing the complexity by a factor four applying phases for groups of 4 states.

As an example, the circuits for the trace gate for $n=2,3,4$ can be realized with
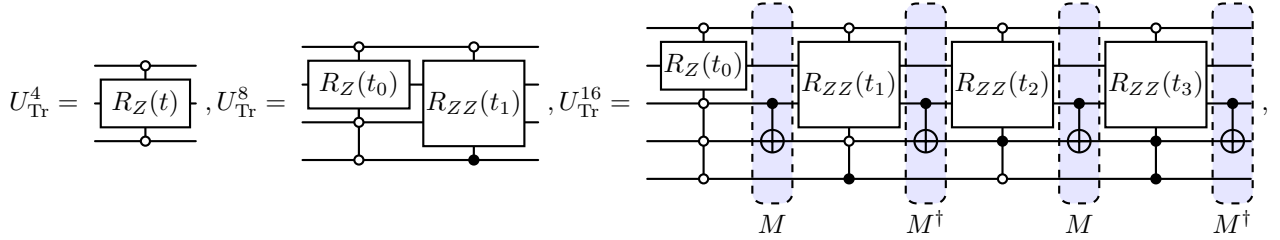
\begin{equation*}
    U_{\Tr}^4  = 
    \begin{quantikz}[row sep={0.5cm,between origins},column sep=.5ex]
         & \octrl{1}      & \\
         & \gate{\RZ (t)} & \\
         & \octrl{-1}     & 
    \end{quantikz},
    U_{\Tr}^8 = 
    \begin{quantikz}[row sep={0.5cm,between origins},column sep=.5ex, wire types = {n,n,n,n}]
        & \octrl{1}       & \octrl{1}       & \arrow[arrows]{lll}\\
        & \gate{\RZ(t_0)} & \gate[wires=2, style={inner xsep=-3pt}]{\RZZ(t_1)} & \arrow[arrows]{lll}\\
        & \octrl{-1}      &                 & \arrow[arrows]{lll}\\
        & \octrl{-2}      & \ctrl{-2}       & \arrow[arrows]{lll}
    \end{quantikz},
    U_{\Tr}^{16} = 
    \begin{quantikz}[row sep={0.5cm,between origins},column sep=.5ex, wire types = {n,n,n,n,n}]
        & \octrl{1}       &
                            \gategroup[5,steps=1,style={dashed,rounded corners,fill=blue!10, inner xsep=-3pt},background, label style={label position=below, yshift=-0.5cm}]{
                            $M$
                            }
                                   & \octrl{1}       & 
                            \gategroup[5,steps=1,style={dashed,rounded corners,fill=blue!10, inner xsep=-3pt},background, label style={label position=below, yshift=-0.5cm}]{
                            $M^\dag$
                            }
                                                                &\octrl{1}      &
                            \gategroup[5,steps=1,style={dashed,rounded corners,fill=blue!10, inner xsep=-3pt},background, label style={label position=below, yshift=-0.5cm}]{
                            $M$
                            }
                                                                                         &\octrl{1}       & 
                            \gategroup[5,steps=1,style={dashed,rounded corners,fill=blue!10, inner xsep=-3pt},background, label style={label position=below, yshift=-0.5cm}]{
                            $M^\dag$
                            }
                                                                                                                   & \arrow[arrows]{lllllllll}\\
        & \gate[style={inner xsep=-3pt}]{\RZ(t_0)} &        & \gate[wires=2, style={inner xsep=-3pt}]{\RZZ(t_1)} &          &\gate[wires=2, style={inner xsep=-3pt}]{\RZZ(t_2)}&        &\gate[wires=2, style={inner xsep=-3pt}]{\RZZ(t_3)} &        & \arrow[arrows]{lllllllll}\\
        & \octrl{-1}      &\ctrl{1}&                 &\ctrl{1}  &               &\ctrl{1}&                &\ctrl{1}& \arrow[arrows]{lllllllll}\\
        & \octrl{-1}      &\targ{} & \octrl{-2}      &\targ{}   &\ctrl{-2}      &\targ{} &\ctrl{-2}       &\targ{} & \arrow[arrows]{lllllllll}\\
        & \octrl{-1}      &        & \ctrl{-1}       &          &\octrl{-1}     &        &\ctrl{-1}       &        & \arrow[arrows]{lllllllll}
    \end{quantikz},
\end{equation*}
where $\RZZ$ can be realized with two CX gates, see Figure~\ref{fig:Pauliexp}.

\subsection{Oracle for \texorpdfstring{\maxkcut{k}}{MAX k-CUT}}
\label{sec:oraclemaxkcut}
Another example is the realization of the oracle operator for the \maxkcut{k} problem.
Given a weighted undirected graph $G = (V, E)$, the \maxkcut{k} problem seeks a partition of the vertex set $V$ into $k$ subsets that maximizes the total weight of edges connecting vertices in different subsets.
By assigning a label $x_i \in {1, \dots, k}$ to each vertex $i \in V$, the \maxkcut{k} cost function can be expressed as
\begin{equation*}
    \underset{\mathbf{x}\in\{1,\dots,k\}^n}{\max} C(\mathbf{x}), \qquad  C(\mathbf{x}) = \sum_{(i,j)\in E} 
    \begin{cases}
        w_{ij}, & \text{ if } x_i \neq x_j\\
        0, & \text{otherwise},
    \end{cases}
\end{equation*}
where $w_{ij}>0$ is the weight of the edge $(i,j)\in E$.
After encoding the labels into $L_k \coloneq \lceil log_2(k)\rceil$ qubits the resulting oracle can be written as the sum of local diagonal projectors, i.e.
$
    H_P = \sum_{e\in E} w_{e} H_{e},
    $
where each local term can be expressed as
\begin{equation*}
    H_{e} = 2 \widehat H_e - I, \ \widehat H_e 
     = \sum_{(i,j) \in \clr } \ket{i}\bra{i} \otimes \ket{j}\bra{j},
\end{equation*}
where 
$\clr = \{\text{sets of equivalent colors}\}$.
Notice that $\widehat H_e$ contains $\Order{k\max_{k}(|\clr_k|)}$ diagonal projection operators of the form $\ket{ij}\bra{ij}$
and can be implemented using the Theorem~\ref{theorem:main} for $\sigma=I$ resulting in a scaling of $\Order{nk\max_{k}(|\clr_k|)}$ X, and CX gates and $\Order{nk\max_{k}(|\clr_k|)+log(k\max_{k}(|\clr_k|)log(1/\epsilon)}$ T gates.

We provide an example we construct the $\widehat H_e$ for the \maxkcut{3} problem.
Defining $\clr=\{ \{0,0\},  \{1,1\},  \{2,2\},  \{3,3\},\allowbreak  \{2,3\},  \{3,2\} \}$ we 
see that $\widehat H_e = P_B$ for $B = \{\ket{0000}, \ket{0101}, \ket{1010}, \ket{1111}, \ket{1011}, \ket{1110} \}$ or in integer representation $ B = \{ \ket{0}, \ket{5}, \ket{10}, \ket{15}, \ket{11}, \ket{14} \}$.
Applying Theorem~\ref{theorem:main} directly gives the following circuit
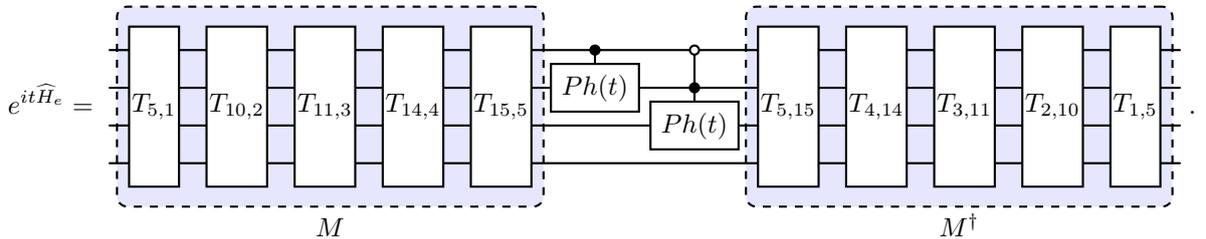
\begin{equation*}
    e^{i t \widehat H_e} = 
    \begin{quantikz}[row sep={0.5cm,between origins},column sep=1ex, wire types = {n,n,n,n}]
        & \gate[wires=4, style={inner xsep=-3pt}]{T_{5,1}}
            \gategroup[4,steps=5,style={dashed,rounded corners,fill=blue!10, inner xsep=-2pt},background, label style={label position=below, yshift=-0.5cm}]{
            $M$
            }
          & \gate[wires=4, style={inner xsep=-3pt}]{T_{10,2}} 
            & \gate[wires=4, style={inner xsep=-3pt}]{T_{11,3}} 
              & \gate[wires=4, style={inner xsep=-3pt}]{T_{14,4}} 
                & \gate[wires=4, style={inner xsep=-3pt}]{T_{15,5}} 
                  & \ctrl{1}
                      & \octrl{1}
                                  & \gate[wires=4, style={inner xsep=-3pt}]{T_{5,15}} 
                                    \gategroup[4,steps=5,style={dashed,rounded corners,fill=blue!10, inner xsep=-2pt},background, label style={label position=below, yshift=-0.5cm}]{
                                    $M^\dag$
                                    }
                                    & \gate[wires=4, style={inner xsep=-3pt}]{T_{4,14}} 
                                      & \gate[wires=4, style={inner xsep=-3pt}]{T_{3,11}} 
                                        & \gate[wires=4, style={inner xsep=-3pt}]{T_{2,10}} 
                                          & \gate[wires=4, style={inner xsep=-3pt}]{T_{1,5}}
                                            & \arrow[arrows]{lllllllllllll}\\
        & & & & & &  \gate{\Ph(t)}    & \ctrl{1} &  & & & & & \arrow[arrows]{lllllllllllll}\\
        & & & & & &     &  \gate{\Ph(t)}         &  & & & & & \arrow[arrows]{lllllllllllll}\\
        & & & & & &
                    & 
                                 &  & & & & & \arrow[arrows]{lllllllllllll}
    \end{quantikz}.
\end{equation*}
By looking at the states carefully, we can apply transpose the states $\ket{0}$ to $\ket{12}$ and  $\ket{5}$ to $\ket{14}$ to arrive the circuit
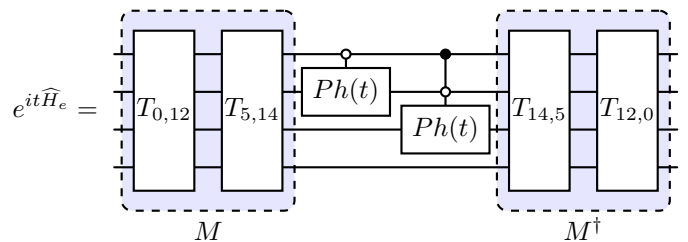
\begin{equation*}
    e^{i t \widehat H_e} = 
    \begin{quantikz}[row sep={0.5cm,between origins},column sep=1ex, wire types = {n,n,n,n}]
        & \gate[wires=4, style={inner xsep=-3pt}]{T_{0,12}}
            \gategroup[4,steps=2,style={dashed,rounded corners,fill=blue!10, inner xsep=-2pt},background, label style={label position=below, yshift=-0.5cm}]{
            $M$
            }
          & \gate[wires=4, style={inner xsep=-3pt}]{T_{5,14}} 
                  & \octrl{1}
                      & \ctrl{1}
                                  & \gate[wires=4, style={inner xsep=-3pt}]{T_{14,5}} 
                                    \gategroup[4,steps=2,style={dashed,rounded corners,fill=blue!10, inner xsep=-2pt},background, label style={label position=below, yshift=-0.5cm}]{
                                    $M^\dag$
                                    }
                                    & \gate[wires=4, style={inner xsep=-3pt}]{T_{12,0}} 
                                            & \arrow[arrows]{lllllll}\\
        & & &  \gate{\Ph(t)}     & \octrl{1} &  & & \arrow[arrows]{lllllll}\\
        & & &   &  \gate{\Ph(t)}         &  & & \arrow[arrows]{lllllll}\\
        & & & 
                    & 
                                 &  & & \arrow[arrows]{lllllll}
    \end{quantikz}.
\end{equation*}

Following Remark~\ref{remark:directsumgg} an alternative is to divide $B$ into two sets that are generated by Pauli X-orbits, e.g.
$B_2 = \langle X_2X_4\rangle \ket{1011}$
and $B_1 = \langle X_1X_3, X_2X_4 \rangle \ket{0000}$.
We can then realize the oracle equivalently with the circuit 
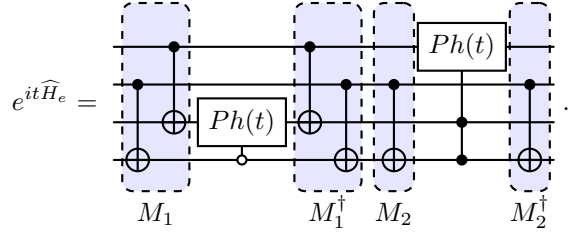
\begin{equation*}
    e^{i t \widehat H_e} = 
    \begin{quantikz}[row sep={0.5cm,between origins},column sep=1ex, wire types = {n,n,n,n}]
        & 
            \gategroup[4,steps=2,style={dashed,rounded corners,fill=blue!10, inner xsep=-2pt},background, label style={label position=below, yshift=-0.5cm}]{
            $M_1$
            }
            & \ctrl{2}
                & 
                    & \ctrl{2}
                        \gategroup[4,steps=2,style={dashed,rounded corners,fill=blue!10, inner xsep=-2pt},background, label style={label position=below, yshift=-0.5cm}]{
                        $M_1^\dag$
                        }
                        &
                        &
                            &
                            \gategroup[4,steps=1,style={dashed,rounded corners,fill=blue!10, inner xsep=-2pt},background, label style={label position=below, yshift=-0.5cm}]{
                            $M_2$
                            }
                                & \gate{\Ph(t)}
                                    &
                                        \gategroup[4,steps=1,style={dashed,rounded corners,fill=blue!10, inner xsep=-2pt},background, label style={label position=below, yshift=-0.5cm}]{
                                        $M_2^\dag$
                                        }
                                        & \arrow[arrows]{llllllllll}\\
        & \ctrl{2}
            & 
                &
                    &
                        & \ctrl{2}
                        &
                            & \ctrl{2}
                                &
                                    & \ctrl{2}
                                        & \arrow[arrows]{llllllllll}\\
        &
            & \targ{} 
                & \gate{\Ph(t)}
                    & \targ{}
                        &
                        &
                            &
                                & \ctrl{-2}
                                    &
                                        & \arrow[arrows]{llllllllll}\\
        & \targ{} 
            &
                & \octrl{-1}
                    &
                        & \targ{}
                        &
                            & \targ{}
                                & \ctrl{-1}
                                    & \targ{}
                                            & \arrow[arrows]{llllllllll}
    \end{quantikz}.
\end{equation*}
Comparing this with the circuit from~\cite{fuchs2021efficient} for $k=3$ shows a drastic improvement in usage of ancilla qubits and gate counts.

\subsection{Mixers for Constrained Optimization}
Here, we want to show some examples of how to use the formalism introduce previously to construct constrained LX-mixers~\cite{fuchs2023LX}.
In this case we want construct mixing unitaries of the form
\begin{equation*}
    U_M(t) = e^{-i t H_M},
    \quad H_M = \sum_{j<k} (T)_{j,k} H_{z_j\leftrightarrow z_k},
    \quad H_{z_j\leftrightarrow z_k} = \ket*{z_j}\!\!\bra*{z_k} + \ket*{z_k}\!\!\bra*{z_j}.
\end{equation*}
The Hamiltonian $H_M$ is called valid~\cite{hadfield2019quantum} if the graph $G_T$ of the adjacency matrix $T$ is undirected and connected.
Given a feasible set $B$ we define the family of graphs $(G_{\lX})_{\lX \in \{I,X\}^n\setminus\{I\}^n}$ where $G_{\lX} = (B, E_{\lX})$.
This gives rise to a family of mixers $(H_{\lX})_{\lX \in \{I,X\}^n\setminus\{I\}^n}$ 
\begin{equation*}
H_{\lX}=
    \sum_{j<k} (T_{\lX})_{j,k} H_{z_j\leftrightarrow z_k} = \sum_{\{\ket{x},\ket{y}\}\in E_{\lX}} H_{x\leftrightarrow y}=\lX P_{V_{\lX}},
\end{equation*}
where $T_{\lX}$ is the adjacency matrix of the graph $G_{\lX}$. As we can see each term in $H_{\lX}$ correspond to an Hamiltonian of the form $\sigma P$ and so we can use the main result in the case $\sigma=X$ for implementing the real time evolution.

\subsubsection{Mixer for subspace generated by Pauli X-orbits}
In the specific case where the subspace generated by a Pauli X-orbit, i.e.
$B=\langle \lX_1, \cdots, \lX_k\rangle\ket{\z}$, for some reference state $\ket{\z}$, it is interesting to notice that the graph $G=\bigcup_{i=1}^kG_{\lX_i}$ is related to the valid adjacency of a $k$-regular graph with $2^k $ vertices.
In fact the graph is has a number of vertices equal to the dimension of the generated subspace $2^k$ and each vertex is connected to exactly k other vertices by the generators of the subgroup $\lX_i$. In particular, since the subspace is generated by a Pauli X-orbit, each generator can be propagated through the basis change M resulting in the following mixer:
\begin{equation*}
    U_M(t)
    = e^{it \sum_{j=1}^{k}\lX_j P_B}
    = e^{it\sum_{j=1}^kM\Tilde{\lX_j}\otimes\ket{0}\bra{0}_{n-k}M^{\dag}}
    = M\prod_{j=1}^{k}C_{\mathbf{0}}^{\c}e^{it\lX_j}M^{\dag}, \quad \c = (n,\cdots, n-k+1)
\end{equation*}

As an example we want to realize a mixing operator for the subspace given by $B=\{\ket{0000}, \ket{1010},\allowbreak \ket{0111},\allowbreak \ket{1101}\} = \langle X_1X_3, X_2X_3X_4\rangle \ket{0000}$.
We note that $P_B=1/4\langle Z_2Z_4, Z_1Z_3Z_4\rangle$.
Defining $\lX_1 = X_1X_3, \lX_2=X_2X_3X_4$ and following the algorithm in the proof of Theorem~\ref{theorem:GroupGeneratedCase}, we can realize
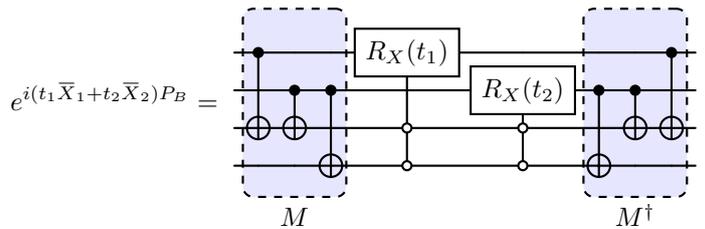
\begin{equation*}
    e^{i (t_1 \lX_1 + t_2 \lX_2) P_B} = 
    \begin{quantikz}[row sep={0.5cm,between origins},column sep=1ex]
    \qw & \ctrl{2} 
    \gategroup[4,steps=3,style={dashed,rounded corners,fill=blue!10, inner xsep=-2pt},background, label style={label position=below, yshift=-0.5cm}]{
    $M$
    }
                    &          &          &\gate{\RX(t_1)} &                 & 
    \gategroup[4,steps=3,style={dashed,rounded corners,fill=blue!10, inner xsep=-2pt},background, label style={label position=below, yshift=-0.5cm}]{
    $M^\dag$
    }
                                                                                        &          & \ctrl{2}  & \\
    \qw &           & \ctrl{1} & \ctrl{2} &                & \gate{\RX(t_2)} & \ctrl{2} & \ctrl{1} &           & \\
    \qw & \targ{}   & \targ{}  &          &\octrl{-2}      & \octrl{-1}      &          & \targ{}  & \targ{}   & \\
    \qw &           &          & \targ{}  &\octrl{-1}      & \octrl{-1}      & \targ{}  &          &           & 
    \end{quantikz}
    ,
\end{equation*}
If we on the other hand apply the LX-mixer~\cite{fuchs2023LX} to the problem, we realize $\lX_1$ and $\lX_2$ with the time evolution of four Pauli terms each, which leads to 8 $\RZ$ gates.

\subsubsection{LX-mixer }
Next, we want to create a mixer for the subspace $B=\{\ket{0000},\ket{1000},\ket{0100},\ket{1100},\ket{0010},\ket{1010}\}$ mixing the pairs connected by $\sigma=X_1$
Let's start by noticing that $P_B=\mathbb{I}\otimes\sum_{z'\in B'}\ket{z'}\bra{z'}$ where $B'=\{\ket{000},\ket{100},\ket{010}\}$ since we have all pairs of states that differ only on the first bit.
We want to map this 3 quantum states to be the first 3 states of the computational basis so we need then to apply the transposition gate $T_{100,001}$.
In this new basis we can implement the evolution circuit for $e^{-itH}$ using the low-pass control operator resulting in the following circuit
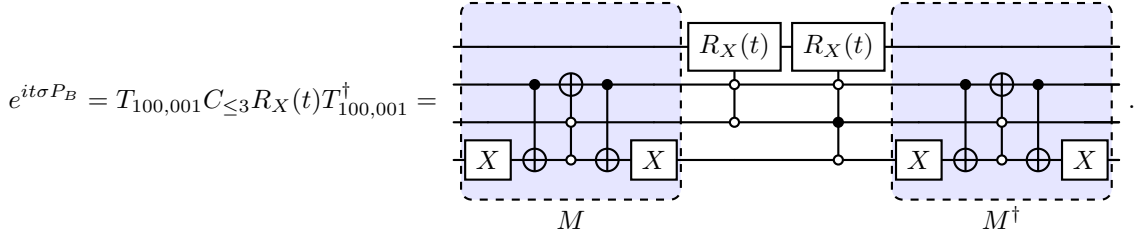
\begin{equation*}
e^{it\sigma P_B}=T_{100,001}C_{\leq 3}\RX(t)T_{100,001}^{\dag}=
\begin{quantikz}[row sep={0.5cm,between origins},column sep=1ex]
\qw & 
    \gategroup[4,steps=5,style={dashed,rounded corners,fill=blue!10, inner xsep=-2pt},background, label style={label position=below, yshift=-0.5cm}]{
    $M$
    }
               &         &          &        &        & \gate{\RX(t)}&\gate{\RX(t)}& 
    \gategroup[4,steps=5,style={dashed,rounded corners,fill=blue!10, inner xsep=-2pt},background, label style={label position=below, yshift=-0.5cm}]{
    $M^\dag$
    }
                                                                                   &        &          &         && \qw \\
\qw &          & \ctrl{2}&\targ{}   &\ctrl{2}&        & \octrl{-1}   & \octrl{-1}  &        &\ctrl{2}&\targ{}   & \ctrl{2}&& \qw \\
\qw &          &         &\octrl{-1}&        &        & \octrl{-1}   & \ctrl{-1}   &        &        &\octrl{-1}&         && \qw \\
\qw & \gate{X} & \targ{} &\octrl{-1}&\targ{} &\gate{X}&              & \octrl{-1}  &\gate{X}&\targ{} &\octrl{-1}& \targ{} &\gate{X}& \qw 
\end{quantikz}.
\end{equation*}
An alternative is to divide $B$ into subspaces that are Pauli X-orbit generated, e.g. into $B_1 = \langle X_1, X_2 \rangle \ket{0000}$
and $B_2 =  \langle X_1 \rangle \ket{0010}$, resulting in
$P_{B_1} = \langle Z_3, Z_4 \rangle$, and 
$P_{B_2} = \langle Z_2, -Z_3, Z_4 \rangle$.
Therefore, the Hamiltonian can also be generated with
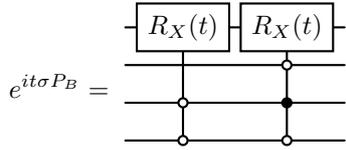
\begin{equation*}
    e^{it\sigma P_B} =
    \begin{quantikz}[row sep={0.5cm,between origins},column sep=1ex]
    & \gate{\RX(t)} & \gate{\RX(t)} & \\
    &               & \octrl{-1}    & \\
    & \octrl{-2}    & \ctrl{-1}     & \\
    & \octrl{-1}    & \octrl{-1}    & \\
    \end{quantikz}.
\end{equation*}
On the other hand, using the method from~\cite{fuchs2023LX}, we can realize $\lX$ with the time evolution of $12$ Pauli terms and equally many $\RZ$ gates.

\section{Conclusion}\label{sec:conclusion}
We have presented a constructive and resource-efficient method to implement the real-time evolution of Hamiltonians of the form $\sigma P_B$, emphasizing compactness with respect to non-transversal gates. Our approach provides significant improvements for both general and Pauli X-orbit generated subspaces, leading to practical circuits applicable across variational quantum algorithms. Notably, we demonstrate how standard quantum operations, such as fermionic excitations and constrained mixers, can be realized with reduced T-gate and ancilla requirements. These findings contribute towards more scalable implementations on fault-tolerant quantum architectures. In future work we plan to apply these methods to suitable problems.

\section{Author contributions}
Franz G. Fuchs and Ruben Pariente Bassa formulated the concept, developed the methodology, made the formal analysis and investigation, wrote the article, made the visualizations in equal parts.

\section{Acknowledgment}
We would like to thank for funding of the work by the Research Council of Norway through project number 332023.

\bibliography{references}

\end{document}